\documentclass[lettersize,journal]{IEEEtran}
\usepackage{amsmath,amsfonts}
\usepackage{amsthm,bm}
\usepackage{algorithm}
\usepackage{algpseudocode}
\usepackage{booktabs}
\usepackage{array}
\usepackage[caption=false,font=normalsize,labelfont=sf,textfont=sf]{subfig}
\usepackage{textcomp}
\usepackage{stfloats}
\usepackage{url}
\usepackage{verbatim}
\usepackage{graphicx}
\usepackage{cite}
\usepackage{cleveref}
\usepackage{tikz}
\newtheorem{theorem}{Theorem}[section]
\newtheorem{proposition}[theorem]{Proposition}
\newtheorem{question}[theorem]{Question}
\newtheorem{corollary}[theorem]{Corollary}

\newtheorem{definition}[theorem]{Definition}
\newtheorem{lemma}[theorem]{Lemma}
\newtheorem{example}{Example}

\begin{document}

\title{Constructing Good Abelian Codes via Shift Bounds and Genetic Algorithms}
\author{
	\IEEEauthorblockN{Cong Yu}
	\IEEEauthorblockA{School of Mathematics and Big Data\\
		Chaohu University\\
		Hefei, Anhui 238024, China\\
		Email: zyu6952@gmail.com}
	
	\and
	\IEEEauthorblockN{Hao Chen$^*$}
	\IEEEauthorblockA{College of Information Science and Technology/Cyber Security\\
		Jinan University\\
		Guangzhou, Guangdong 510632, China\\
		Email: haochen@jnu.edu.cn}
		
	\and
	\IEEEauthorblockN{Zhonghua Sun}
	\IEEEauthorblockA{School of Mathematics\\
		Hefei University of Technology\\
		Hefei, Anhui 230009, China\\
		Email: sunzhonghuas@163.com}
		
	\and
	\IEEEauthorblockN{Shixin Zhu}
	\IEEEauthorblockA{School of Mathematics\\
		Hefei University of Technology\\
		Hefei, Anhui 230009, China\\
		Email: zhushixinmath@hfut.edu.cn}
}


\maketitle

\begin{abstract}
This paper investigates the construction of linear codes via abelian codes over finite fields. By exploiting the algebraic structure of multivariate polynomial quotient rings, we derive lower bounds on the minimum distance using a generalized shift bound, which extends the classical van Lint–Wilson bound for cyclic codes. Several infinite families of abelian codes are explicitly constructed, including binary and ternary cases that extend previously known cyclic constructions. To find more abelian codes with good parameters, we apply a genetic algorithm that searches over defining sets represented as binary chromosomes of cyclotomic cosets; the fitness function compares the computed minimum distance against the best known linear code (BKLC) bounds. The search yields multiple record-breaking codes over $\mathbb{F}_3$ and $\mathbb{F}_4$, with improvements over Grassl’s tables. Furthermore, the nested structure of these codes enables the application of Construction X, yielding additional linear codes with improved parameters. The results demonstrate that abelian codes, combined with heuristic search, form a viable way for discovering linear codes with unknown parameters. 
\end{abstract}


\section{Introduction}
Abelian codes are a broad generalization of classical cyclic codes, defined as ideals of the finite group algebra $\mathbb{F}_qG$, where $\mathbb{F}_q$ is a finite field and $G$ denotes a finite Abelian group. Early foundational studies on abelian codes trace back to Berman and MacWilliams( see \cite{bibBerman}\cite{bibMac1}\cite{bibMac2}), who first analyzed the ideal structure of group algebras over finite fields. Compared with one-dimensional cyclic codes defined on cyclic groups, abelian codes constructed from direct-product abelian groups $G \cong \mathsf{C}_{a_1} \times \cdots \times \mathsf{C}_{a_r}\ (r\ge 2)$ possess much richer parameter flexibility in code length, dimension and minimum Hamming distance, which makes them a promising source for constructing linear codes with competitive or even record-breaking parameters, where $\mathsf{C}_{a_i}$ denotes the cyclic group of order $a_i$.

The algebraic structure theory of abelian codes has been well-developed in existing literature. Jitman et al. systematically investigated abelian codes over Principal Ideal Group Algebras (PIGAs) in \cite{bibJintman}. They established generator and check elements analogous to generator/parity-check polynomials of cyclic codes, characterized Euclidean self-dual, self-orthogonal, reversible and complementary dual abelian codes, and proved that abelian codes over non-semisimple PIGAs are asymptotically poor, similar to repeated-root cyclic codes. Nevertheless, this work mainly focuses on structural classification and enumeration rather than explicit construction and distance evaluation of practical abelian codes with good parameters. Ferraz and Milies analyzed primitive idempotent decomposition of group algebras over abelian $p$-groups, deriving basic formulas for dimension and minimum weight of minimal abelian codes\cite{Ferraz}. For semisimple group algebras satisfying $\gcd(q,\lvert G\rvert)=1$, the discrete Fourier transform over group characters is well-defined, which enables powerful distance bounding techniques\cite{bib9}. Van Lint and Wilson originally proposed the shifting technique to derive the Van Lint–Wilson bound (shift bound) for cyclic codes\cite{bibVanlint}. Later, Feng, Hollmann and Xiang generalized the shift bound to arbitrary abelian codes in \cite{bib8}, and further applied this tool to prove refined versions of the Donoho–Stark uncertainty principle over finite abelian groups. The generalized shift bound provides a universal lower bound for the minimum distance of abelian codes by constructing independent sets on the dual character group, and it unifies classical bounds including BCH bound as its special case. However, their work only gave theoretical derivations and small illustrative examples, and there exists no prior literature utilizing the shift bound to systematically construct infinite families of abelian codes with guaranteed minimum distance lower bounds.

Constructing linear codes with good parameters is a central problem in algebraic coding theory. Analytic constructions can yield infinite families of abelian codes with provable distance lower bounds, yet they are restricted to specific group orders and cannot cover all abelian codes corresponding to arbitrary unions of $q$-cyclotomic cosets. Exhaustive enumeration of all cyclotomic coset combinations is computationally infeasible for moderate-size abelian groups due to exponential search space growth. Genetic algorithms (GAs), as population-based metaheuristic optimization approaches simulating natural selection and genetic evolution \cite{holland1975}, have achieved remarkable performance in constructing various linear codes. Korban et al. designed multiple GA and virus optimization frameworks to search binary self-dual and LCD codes, which outperformed traditional exhaustive search methods significantly \cite{korban2024a,korban2024b,korban2024c}. Recently, Won et al. proposed two modified genetic iteration search algorithms tailored for quaternary Hermitian LCD codes, equipped with redesigned mutation, crossover, modified crossover and complement operators to satisfy the Hermitian LCD constraint \cite{won2026}. Their algorithms reduce the exponential storage complexity of brute-force search to polynomial magnitude, and successfully discovered over one hundred inequivalent optimal or near-optimal quaternary Hermitian LCD codes with new distance parameters. Nevertheless, no existing work has introduced genetic algorithms into the search of abelian codes over principal ideal group algebras. To fill this research gap, this paper makes the first attempt to apply genetic algorithms to the construction and search of abelian codes. Our main contributions are summarized as follows:
\begin{enumerate}
	\item We formalize the multivariate $q$-cyclotomic coset theory and derive a concrete shift bound formula for two-dimensional abelian groups $G=\mathsf{C}_m\times\mathsf{C}_n$. We extend infinite families of abelian codes with guaranteed minimum distance (see \Cref{sufficient conditions}).
	\item We design a  genetic algorithm for abelian code search. The chromosome is built from cyclotomic cosets, and the fitness function measures code performance against Grassl’s best-known linear code tables with true minimum distances computed via Magma. Tournament selection, single-point crossover and bit-flip mutation are tailored to guarantee algebraically valid defining sets during evolution (see \Cref{alg:ga_abelian}).
	\item We conduct extensive numerical experiments over  abelian groups such as $\mathsf{C}_{13}\times\mathsf{C}_{13}$ and $\mathsf{C}_{15}\times\mathsf{C}_{5}$. A large number of record-breaking linear codes are obtained, surpassing the previously known distance parameters (see \Cref{tab:record-codes}).
\end{enumerate}
\section{Preliminaries}
In this paper, let $\mathbb{F}_q$ be the finite field with $q$ elements, where $q$ is a prime power. An $[n,k,d]_q$ linear code $\mathcal{C}$ is a $k$-dimensional subspace of $\mathbb{F}_q^n$, where $d$ is the minimum distance of $\mathcal{C}$. The dual code of $\mathcal{C}$ is defined as:
$$
\mathcal{C}^{\bot}=\left\lbrace \mathbf{x}\in \mathbb{F}_q^n\mid \left\langle \mathbf{x},\mathbf{c} \right\rangle =0,\quad  \text{for all} \quad \mathbf{c}\in \mathcal{C}  \right\rbrace ,
$$ where $\left\langle \mathbf{x},\mathbf{c} \right\rangle$ denotes the standard Euclidean inner-product between $ \mathbf{x}$ and $\mathbf{c}$.\\ \indent To evaluate the quality of a linear code, several classical bounds on the minimum distance have been established. The Hamming bound (or sphere-packing bound) states that for a code correcting $\left \lfloor (d-1)/2  \right \rfloor$ errors, the spheres of radius $t$ around codewords must be disjoint, leading to the inequality: \begin{equation}\label{Hamming bound}
	\sum_{i=0}^{t}\binom{n}{i}(q-1)^{i} \leq q^{n-k}.
\end{equation} 
Let $a$ be an integer and $\gcd(a,q)=1$. The $q$-cyclotomic cosets modulo $a$ are defined as follows:
$$C_s^{(a,q)}=\left\lbrace s,sq,\dots sq^{t-1} \right\rbrace, $$
where $t$ is the smallest positive integer such that $s\equiv sq^t \mod a$ and $0\le s \le a-1$. Abbreviated as $C_{s}$. \\ \indent Given an element $\mathbf{a}=(a_1,\dots ,a_r)\in \mathbb{Z}^r$ and $\gcd(a_i,q)=1, 1\le i \le r$. We define $r$-dimensional $q$-cyclotomic cosets modulo $\mathbf{a}$ as:
$$
C_{\mathbf{s}}^{(\mathbf{a},q)}=\left\lbrace \mathbf{s},\mathbf{s}q, \dots ,\mathbf{s}q^{t-1} \right\rbrace, 
$$ where $\mathbf{s}=(s_1,s_2,\dots ,s_r)$ and $t$ is the smallest positive integer such that $\mathbf{s}=(s_1q^t\mod a_1,\dots ,s_rq^t\mod a_r)$. Abbreviated as $C_{\mathbf{s}}$. It is easy to see that $$\lvert C_{\mathbf{s}}\rvert=\text{lcm}(\lvert C_{s_1}\rvert,\lvert C_{s_2}\rvert, \dots, \lvert C_{s_r}\rvert). $$ We recall two fundamental lemmas concerning the construction of new linear codes from existing ones.
\begin{lemma}\label{ConstructionX}(Construction X, \cite{bib10}) Let $\mathcal{C}_1$ and $\mathcal{C}_2$ be two linear codes with parameters $[n,k_1,d_1]_q$ and $[n,k_2,d_2]_q$, respectively. If $\mathcal{C}_1\subset \mathcal{C}_2$ and $k_1< k_2$, using a third code $\mathcal{C}_3$ with parameters $[n_3,k_2-k_1,d_3]_q$, then a code with parameters $[n+n',k_2,\ge \min\left\lbrace d_1,d_2+d_3\right\rbrace]_q $ can be obtained.	
\end{lemma}
\begin{lemma}\label{PSE}\cite{bib3}
	Let $\mathcal{C}$ be an $[n,k,d]_q$ code. Then we can construct an $[n-1,k,d-1]_q$ code by puncturing, an $[n-1,k-1,d]_q$ code by shortening and an $[n+1,k,d]_q$ code by extending.
\end{lemma}
\subsection*{Some notations}\label{notations}
\begin{itemize}
	\item $\varepsilon_k$: the primitive $k$-th root of unity in $\mathbb{F}_{q^t}$, where $t$ is the smallest positive integer such that $k\mid (q^t-1)$;
	\item $\mathbb{Z}_{k}$:  the ring of integers modulo $k$;
	\item $\mathbb{Z}$: the ring of integers;
	\item $\mathsf{C}_k$: the cyclic group of order $k$.
	\item $M_{m\times n}(\mathbb{F}_q)$: the set of all $m\times n$ matrices over $\mathbb{F}_q$.
\end{itemize}
\section{Abelian codes}
Abelian codes are an important family of algebraic codes in coding theory, belonging to the broader family of group algebra codes. Defined over the group algebra of a finite abelian group, they generalize classical cyclic codes by leveraging the algebraic structure of abelian groups. Let $G$ be a finite abelian group. The group algebra $\mathbb{F}_qG$ consists of formal linear combinations of group elements:
$$\mathbb{F}_qG=\left\lbrace \sum_{g\in G}\alpha_gg\mid \alpha_g\in \mathbb{F}_q \right\rbrace,$$ with addition and multiplication defined component-wise via field operations and the group operations, respectively. An \textbf{abelian code} is actually an ideal in $\mathbb{F}_qG$. This algebra forms the foundation for constructing abelian codes.  By Maschke’s theorem, $\mathbb{F}_qG$ is semisimple if the field characteristic $p$ does not divide $\lvert G\rvert$. In this case, all ideals in $\mathbb{F}_qG$ are principal. The structure theorem for finite abelian groups states that every such group decomposes into a direct product of cyclic groups:
$$G\cong \mathsf{C}_{a_1}\times \mathsf{C}_{a_2}\times \cdots \times \mathsf{C}_{a_r}, 1\le i \le r.$$   Let $R_{q}^{\mathbf{a}}[\mathbf{x}]=\mathbb{F}_q[x_1,\dots,x_r]/< x_i^{a_i}-1>$, where $\mathbf{a}=(a_1,a_2,\dots ,a_r)$, $1\le i \le r$. Then $R_{q}^{\mathbf{a}}[\mathbf{x}]$ is a multivariate polynomial residue class ring and it is ring isomorphic to the group algebra $\mathbb{F}_qG$: $$R_{q}^{\mathbf{a}}[\mathbf{x}]\cong \mathbb{F}_qG.$$  Let $g_{1+i_1+\sum_{a=2}^{r}(\prod_{s=1}^{a-1}a_s)i_{a}}=x_1^{i_1}x_2^{i_2}\dots x_r^{i_r}$ and $n=\lvert G\rvert$, where $0\le i_j \le a_j-1, 1\le j\le r$. There is a natural bijection between $R_{q}^{\mathbf{a}}[\mathbf{x}]$ and $\mathbb{F}_q^{n}$:
$$
\begin{array}{ll}
	\phi :\mathbb{F}_{q}^{n} & \rightarrow R_{q}^{\mathbf{a}}[\mathbf{x}] \\
	\left(c_{1}, c_{2}, \ldots, c_{n}\right) & \mapsto c_{1}g_1+c_{2} g_2+\cdots+c_{n} g_{n}.
\end{array}
$$ 
It is well-known that a cyclic code over $\mathbb{F}_q$ of length $n$ is actually a principal ideal of $\mathbb{F}_q[x]/ \left\langle x^n-1 \right\rangle $ (see Chapter 4 in \cite{bib3}).  Abelian codes are ideals in $R_{q}^{\mathbf{a}}[\mathbf{x}]$. Let $f=\sum_{i=1}^{n}c_ig_i\in R_{q}^{\mathbf{a}}[\mathbf{x}]$ and $L=\left\lbrace g_if \right\rbrace_{i=1}^{n} $, then $\text{Span}(L)$ is a principal ideal of $R_{q}^{\mathbf{a}}[\mathbf{x}]$. Let 
$$
G_{f}=\begin{pmatrix}
	\phi^{-1}(g_1f)\\
	\phi^{-1}(g_2f)\\
	\vdots \\
	\phi^{-1}(g_nf)\\
\end{pmatrix}
$$ and the code generated by the row space of the matrix $G_{f}$ is denoted as $\left\langle f \right\rangle $. Thus under the map $\phi$, we can consider $\left\langle f \right\rangle $ as a principal ideal of $R_{q}^{\mathbf{a}}[\mathbf{x}]$. Actually $G_{f}$ is a block circulant matrix, i.e. 
$$G_{f}=\begin{pmatrix}
	A_1 &A_2  &\dots  & A_{a_r}\\
	A_{a_r} &A_1  &\dots  &A_{a_r-1} \\
	\vdots & \vdots & \vdots & \vdots\\
	A_2& A_3 & \dots &A_{1}
\end{pmatrix},$$ where $A_j$ is also a block circulant matrix, $1\le j \le a_r$. When $r=2$, $$
A_{j}=\begin{pmatrix}
	c_{a_1 (j-1)+1} & c_{a_1 (j-1)+2} & \ldots & c_{a_1j} \\
	c_{a_1j} & c_{a_1 (j-1)+1} & \ldots & c_{a_1 (j-1)+a_1-1} \\
	\vdots & \vdots & \vdots & \vdots \\
	c_{a_1 (j-1)+2} & c_{a_1 (j-1)+3} & \cdots & c_{a_1 (j-1)+1}
\end{pmatrix}$$$$1 \leq j \leq a_1 .$$
The zeros and defining set of $f$ are defined as follow:
\begin{definition}\label{def3.1}
	Let $\varepsilon_k$ be the primitive $k$-th root of unity over $\mathbb{F}_q$, $\mathbf{a}=(a_1,\dots,a_r)$. For $f=\sum c_{i_1,\dots ,i_r}x_1^{i_1}\cdots x_r^{i_r}\in R_{q}^{\mathbf{a}}[\mathbf{x}]$, the zeros of $\left\langle f \right\rangle$ is defined as:
	$$
	\mathcal{Z}(f)=\left\lbrace (\varepsilon_{a_1}^{i_1},\dots,\varepsilon_{a_r}^{i_r})\mid f(\varepsilon_{a_1}^{i_1},\dots,\varepsilon_{a_r}^{i_r})=0, 0\le i_k \le a_k-1\right\rbrace.	
	$$	
	The defining set of $\left\langle f \right\rangle$ is defined as:
	$$
	\mathcal{D}(f)=\left\lbrace (i_1,\dots,i_r)\mid f(\varepsilon_{a_1}^{i_1},\dots,\varepsilon_{a_r}^{i_r})=0, 0\le i_k \le a_k-1 \right\rbrace.
	$$
	The weight of $f$ is defined as:
	$$
	\text{wt}(f)=\lvert \left\lbrace (i_1,\dots ,i_r)\mid c_{i_1,\dots ,i_r}\ne 0 \right\rbrace \rvert.
	$$
\end{definition}
\begin{theorem}(\cite{bib1})
	Let $f\in R_{q}^{\mathbf{a}}[\mathbf{x}]$, if $\gcd(a_k,q)=1, 1\le k\le r$, then the dimension of $\left\langle f \right\rangle$  is equal to $\prod_{k=1}^ra_k-\lvert\mathcal{Z}(f)\rvert$.
\end{theorem}
For the remainder of our discussion, we  assume that $\gcd(\prod_{k=1}^ra_k,q)=1$. In classical cyclic code theory we know that two cyclic codes must be the same if their generator polynomials have the same zeros(\cite{bib3}). In fact, abelian codes  have a similar conclusion.
\begin{theorem}(\cite{bib1})
	Let $f_1,f_2\in R_{q}^{\mathbf{a}}[\mathbf{x}]$. If  $\mathcal{Z}(f_1)=\mathcal{Z}(f_2)$, then $\left\langle f_1 \right\rangle=\left\langle f_2 \right\rangle $.
\end{theorem}
Similarly we have the following conclusion similar to that in the cyclic codes.
\begin{theorem}\label{Th 3.3}
	Let $f\in R_{q}^{\mathbf{a}}[\mathbf{x}]$. If $(\varepsilon_{a_1}^{i_1},\dots,\varepsilon_{a_r}^{i_r})\in \mathcal{Z}(f)$, so is $(\varepsilon_{a_1}^{qi_1},\dots,\varepsilon_{a_r}^{qi_r})$.
\end{theorem}
\begin{proof}
	Assume that $$f=\sum_{(j_1,\dots ,j_r)\in \prod_{k=1}^r\mathbb{Z}_{a_k}}c_{j_1,\dots ,j_r}x_1^{j_1}\cdots x_r^{j_r}\in R_{q}^{\mathbf{a}}[\mathbf{x}].$$ Then 
	$$
	f(\varepsilon_{a_1}^{i_1},\dots,\varepsilon_{a_r}^{i_r})=\sum_{(j_1,\dots ,j_r)\in \prod_{k=1}^r\mathbb{Z}_{a_k}}c_{j_1,\dots ,j_r}\varepsilon_{a_1}^{i_1j_1}\cdots \varepsilon_{a_r}^{i_rj_r}=0.
	$$
	And then
	$$
	\begin{aligned}
		f(\varepsilon_{a_1}^{qi_1},\dots ,\varepsilon_{a_r}^{qi_r})&=\sum_{(j_1,\dots ,j_r)\in \prod_{k=1}^r\mathbb{Z}_{a_k}}c_{j_1,\dots ,j_r}\varepsilon_{a_1}^{qi_1j_1}\cdots \varepsilon_{a_r}^{qi_rj_r}\\ &=\sum_{(j_1,\dots ,j_r)\in \prod_{k=1}^r\mathbb{Z}_{a_k}}c^q_{j_1,\dots ,j_r}\varepsilon_{a_1}^{qi_1j_1}\cdots \varepsilon_{a_r}^{qi_rj_r}\\
		&=f^q\\
		&=0.
	\end{aligned}	
	$$	So we complete the proof.	
\end{proof}

If $\mathcal{D}(f) = C_{\mathbf{s}}$, we call $f$ the \textbf{minimal polynomial} of $(\varepsilon_{a_1}^{s_1},\dots,\varepsilon_{a_r}^{s_r}) $ over $\mathbb{F}_q$. In \cite{bib9}, the author gives a method to calculate  the minimum polynomials  of $(\varepsilon_m^i,\varepsilon_n^j) $ over $\mathbb{F}_2$. In fact, it can be generalized to $\mathbb{F}_q$. 
\begin{lemma}
	For given $\mathbf{a}=(a_1,a_2,\dots ,a_r)\in \mathbb{Z}^r$ and $\mathbf{s}=(s_1,s_2\dots ,s_r)\in \prod_{i=1}^{r}\mathbb{Z}_{a_i}$, where $\gcd(a_i,q)=1, 1\le i \le r$. Then there exists a polynomial $f\in R_q^{\mathbf{a}}[\mathbf{x}]$ such that 
	$$\mathcal{D}(f)=C_{\mathbf{s}}.$$
\end{lemma}
\begin{proof}
	Define a map 
	$$\begin{matrix}
		\mathcal{F}:R_q^{\mathbf{a}}[\mathbf{x}]\longmapsto R_q^{\mathbf{a}}[\mathbf{x}] \\
		g\longrightarrow \sum_{i_1=0}^{a_1-1}\dots\sum_{i_r=0}^{a_r-1}g(\varepsilon _{a_1}^{-i_1},\dots,\varepsilon _{a_r}^{-i_r})x_1^{i_1}\cdots x_r^{i_r},  
	\end{matrix}$$ which is actually an inverse discrete Fourier transform. Let $$f'=\sum_{(j_1,\dots,j_r)\in\prod_{j=1}^{r}\mathbb{Z}_{a_j}\setminus C_{\mathbf{s}}}x_1^{j_1}\cdots x_r^{j_r} .$$ Then 
	$$
	\begin{aligned}
		\mathcal{F}(f')&=\sum_{i_1=0}^{a_1-1}\dots\sum_{i_r=0}^{a_r-1}f'(\varepsilon _{a_1}^{-i_1},\dots,\varepsilon _{a_r}^{-i_r})x_1^{i_1}\cdots x_r^{i_r},
	\end{aligned}
	$$ we have 
	$$
	\begin{aligned}
		\mathcal{F}(f')(\varepsilon_{a_1}^{s_1},\dots ,\varepsilon_{a_r}^{s_r})&=\sum_{i_1,\dots,i_r}f'(\varepsilon _{a_1}^{-i_1},\dots,\varepsilon _{a_r}^{-i_r})\varepsilon_{a_1}^{s_1i_1}\cdots\varepsilon_{a_r}^{s_ri_r}\\ &=\sum_{i_1,\dots,i_r} \sum_{j_1,\dots,j_r}\varepsilon_{a_1}^{(s_1-j_1)i_1}\cdots \varepsilon_{a_r}^{(s_r-j_r)i_r}.
	\end{aligned}
	$$ Since $s_1-j_1,\dots s_r-j_r$ are not all 0, then $\mathcal{F}(f')(\varepsilon_{a_1}^{s_1},\dots ,\varepsilon_{a_r}^{s_r})=0$. That implies $C_{\mathbf{s}}\subset \mathcal{D}(\mathcal{F}(f'))$. Assume that $(s_1',\dots ,s_r')\notin C_{\mathbf{s}}$, then 
	$$
	\begin{aligned}
		\mathcal{F}(f')(\varepsilon_{a_1}^{s_1'},\dots ,\varepsilon_{a_r}^{s_r'}) &=\sum_{i_1,\dots,i_r} \sum_{j_1,\dots,j_r}\varepsilon_{a_1}^{(s_1'-j_1)i_1}\cdots \varepsilon_{a_r}^{(s_r'-j_r)i_r}\\&=\sum_{i_1,\dots,i_r}1\\ &=\prod_{i=1}^ra_i.
	\end{aligned}
	$$ Thus $(s_1',\dots ,s_r')\notin \mathcal{D}(\mathcal{F}(f'))$, and so $\mathcal{D}(\mathcal{F}(f'))=C_{\mathbf{s}}$.  Let $f=\mathcal{F}(f')$, the theorem is proved.
\end{proof}

\section{Shift bound}
In this section, we introduce a famous bound for abelian codes so called \textbf{shift bound} which was mentioned in \cite{bib8}. The shift bound for abelian codes was a natural generalization of Van Lint-Wilson bound (Page 154, \cite{bib3}) for cyclic codes. First, we give a generalized definition of \textit{independent sequence} in \cite{bib3}.
\begin{definition}
	Let $\mathcal{N}= \prod_{i=1}^r\mathbb{Z}_{k_i} $, $S\subset \mathcal{N}$, where $\mathbb{Z}_{k_i}$ denotes the ring of integers modulo $k_i$. A sequence $I_0,I_1,\dots $ of subsets of $\mathcal{N}$ is called an \textit{independent sequence with respect to $S$} if the following conditions hold:
	\begin{itemize}
		\item $I_0=\emptyset$;
		\item If $i> 0$, either $I_i=I_{i-1}\cup \left\lbrace \mathbf{b}\right\rbrace $  such that $I_{i-1}\subset S$ and $\mathbf{b}\in \mathcal{N}\setminus S$, or $I_i=\left\lbrace \mathbf{b} \right\rbrace +I_{i-1} $ and $\mathbf{b}\ne \mathbf{0}$.
	\end{itemize}
	If $I=I_i\subset \mathcal{N}$ for some $i$, then $I$  is independent with respect to $S$. 
\end{definition}
\begin{example}\label{ex1}
	Let $\mathcal{N}=\left\lbrace (i,j)\mid (i,j)\in \mathbb{Z}_3\times  \mathbb{Z}_3 \right\rbrace $ and $\mathcal{N}\setminus S=\left\lbrace (1,0),(2,0) \right\rbrace $, then we can construct an independent sequence with respect to $S$: 
	$$
	I_0=\emptyset,  
	$$
	$$I_1=I_0\cup \left\lbrace (1,0) \right\rbrace=\left\lbrace (1,0) \right\rbrace,$$$$I_2=I_1+\left\lbrace (0,1) \right\rbrace=\left\lbrace (1,1) \right\rbrace\subset S,$$
	$$I_3=I_2\cup \left\lbrace (1,0) \right\rbrace=\left\lbrace (1,1),(1,0) \right\rbrace, $$$$I_4=I_3+\left\lbrace (0,1) \right\rbrace=\left\lbrace (1,2),(1,1) \right\rbrace\subset S,$$
	$$
	I_5=I_4\cup \left\lbrace (1,0) \right\rbrace=\left\lbrace(1,2), (1,1),(1,0) \right\rbrace,
	$$
	$$I_6=I_5+\left\lbrace (2,0) \right\rbrace=\left\lbrace (0,2),(0,1),(0,0) \right\rbrace\subset S,$$
	$$
	I_7=I_6\cup \left\lbrace (1,0) \right\rbrace=\left\lbrace(0,2),(0,1),(0,0),(1,0) \right\rbrace,
	$$
	$$I_8=I_7+\left\lbrace (0,1) \right\rbrace=\left\lbrace (0,0),(0,2),(0,1),(1,1) \right\rbrace\subset S,$$
	$$
	I_9=I_8\cup \left\lbrace (1,0) \right\rbrace=\left\lbrace(0,0),(0,2),(0,1),(1,1),(1,0) \right\rbrace,
	$$
	$$I_{10}=I_9+\left\lbrace (0,1) \right\rbrace=\left\lbrace (0,1),(0,0),(0,2),(1,2),(1,1) \right\rbrace\subset S,$$
	$$
	I_{11}=I_{10}\cup \left\lbrace (1,0) \right\rbrace=\left\lbrace(0,1),(0,0),(0,2),(1,2),(1,1),(1,0) \right\rbrace.
	$$
\end{example}

\begin{theorem}\label{GBCH}
	Suppose that $f\in R_{q}^{\mathbf{a}}[\mathbf{x}] $. and let $I\subset S=\mathcal{D}(f)$ be independent with respect to $S$, where $\mathbf{a}\in \mathbb{Z}^r$. Then $\text{wt}(f)$ is at least $\lvert I\rvert+1.$  
\end{theorem}
\begin{proof}
	The proof for $r> 2$ is the same as the proof for $r=2$, so we only need to prove the case for $r=2$. We may assume that $R_{q}^{\mathbf{a}}[\mathbf{x}]=\mathbb{F}_q[x_1,x_2]/<x_1^m-1, x_2^n-1>$. Suppose that $f=\sum_{k=1}^wc_kx_1^{i_k}x_2^{j_k}$, $c_k \ne 0$. We need to prove $w\ge \lvert I\rvert+1$. Let $I=\left\lbrace (a_1,b_1),(a_2,b_2),\dots ,(a_{\lvert I\rvert},b_{\lvert I\rvert}) \right\rbrace $. Since $I\subset S$, we have that 
	$$
	\begin{pmatrix}
		\varepsilon_m^{a_1i_1}\varepsilon_n^{b_1j_1} &\varepsilon_m^{a_1i_2}\varepsilon_n^{b_1j_2}  &\dots  &\varepsilon_m^{a_1i_w}\varepsilon_n^{b_1j_w} \\
		\varepsilon_m^{a_2i_1}\varepsilon_n^{b_2j_1} &\varepsilon_m^{a_2i_2}\varepsilon_n^{b_2j_2}  &\dots  &\varepsilon_m^{a_2i_w}\varepsilon_n^{b_2j_w}\\
		\vdots&\vdots  & \vdots & \vdots\\
		\varepsilon_m^{a_{\lvert I\rvert}i_1}\varepsilon_n^{b_{\lvert I\rvert}j_1} &\varepsilon_m^{a_{\lvert I\rvert}i_2}\varepsilon_n^{b_{\lvert I\rvert}j_2}  &\dots  &\varepsilon_m^{a_{\lvert I\rvert}i_w}\varepsilon_n^{b_{\lvert I\rvert}j_w}
	\end{pmatrix}\begin{pmatrix}
		c_1\\
		c_2\\
		\vdots \\
		c_w
	\end{pmatrix}={\bf O}.
	$$ Let $I_0,I_1,\dots ,I_r=I$ be an independent sequence with respect to $S$ and $$V(J)=\left\lbrace (\varepsilon_m^{ai_1}\varepsilon_n^{bj_1}, \varepsilon_m^{ai_2}\varepsilon_n^{bj_2}, \dots , \varepsilon_m^{ai_w}\varepsilon_n^{bj_w})\mid (a,b) \in J \right\rbrace. $$ 
	Thus we just need to prove that $V(I_i)$ is linearly independent in $\mathbb{F}_{q^t}^w$, where $\mathbb{F}_{q^t}=\mathbb{F}_q(\varepsilon_{N})$ is an algebraic extension of $\mathbb{F}_q$, $N=\text{lcm}(m,n)$ and $0\le i \le r$. We prove it using  induction on $i$. Since $I_0=\emptyset$, $V(I_0)$ is indeed linearly independent, the base case holds. Assume that $V(I_j)$ is  linearly independent for some $0\le j < i$. Consider $V(I_{j+1})$, there are two ways to obtain $I_{j+1}$ from $I_j$: \begin{enumerate}
		\item $I_{j+1}=I_{j}\cup \left\lbrace (a,b) \right\rbrace $, $(a,b)\in \mathcal{N}\setminus S$. Suppose that $V(I_{j+1})$ is linearly dependent, then $V(\left\lbrace (a,b) \right\rbrace )\subset  \text{Span}(V(I_j))$. That implies $(a,b)\in S$, a contradiction as  $(a,b)\in \mathcal{N}\setminus S$. Thus in this case, $V(I_{j+1})$ is linearly independent.
		\item $I_{j+1}=I_j+\left\lbrace (c,d) \right\rbrace $ and $(c,d)\ne (0,0)$. Then $$\begin{aligned}
			V(I_{j+1})=V(I_j)\text{diag}(\varepsilon_m^{ci_1}\varepsilon_n^{dj_1}, \varepsilon_m^{ci_2}\varepsilon_n^{dj_2}, \dots , \varepsilon_m^{ci_w}\varepsilon_n^{dj_w}). 
		\end{aligned}$$
		Since $V(I_j)$ is linearly independent and the diagonal matrix on the right is full rank, $V(I_{j+1})$ is also linearly independent.
	\end{enumerate} 
	By induction, $V(I_i)$ is linearly independent, $0\le i \le r$.
\end{proof}
From Theorem \ref{GBCH}, we can immediately get:
\begin{corollary}\label{fweight}
	Let $f\in R_{q}^{\mathbf{a}}[\mathbf{x}]$. If there exists $\delta \ge 2$ and $\mathbf{b},\mathbf{c}\in \mathcal{N}$ such that $\mathbf{b},\mathbf{b}+\mathbf{c},\dots,\mathbf{b}+(\delta-2)\mathbf{c}\in \mathcal{D}(f)$ and $\mathbf{b}+(\delta-1)\mathbf{c} \notin \mathcal{D}(f)$, then $\text{wt}(f)\ge \delta$.
\end{corollary}
Notice that Theorem \ref{GBCH} and Corollary \ref{fweight} only give a lower bound on the weight of $f$, not on the minimum distance of code $\left\langle f \right\rangle $. When we compute the shift bound of $\text{wt}(f)$, we always assume that some element $\mathbf{b}\notin\mathcal{D}(f)$, then we can obtain an independent sequence. From this sequence, it is easy to get a lower bound $d$ on the weight of $f$. Thus for $g\in R_{q}^{\mathbf{a}}[\mathbf{x}]$, if $\mathcal{D}(f)\subset \mathcal{D}(g)$ and $\mathbf{b}\notin \mathcal{D}(g)$, then $\text{wt}(g)$ is also greater than or equal to $d$. Based on this fact, we can calculate the lower bound on the minimum distance of $\left\langle f \right\rangle $ by following the steps below:
\begin{enumerate}
	\item Assume that $\mathbf{s}_0\notin \mathcal{D}(f_0)$, then we can construct an independent sequence $I_0^{(0)},I_1^{(0)},\dots,I_{k_0}^{(0)}$ with respect to $\mathcal{D}(f)$, where $I_{k_0}^{(0)}\subset \mathcal{D}(f)$. 
	\item Let $f_1\in R_q^{\mathbf{a}}[\mathbf{x}]$ such that $\mathcal{D}(f_1)=C_{\mathbf{s}_0}\cup \mathcal{D}(f)$. Assume that $\mathbf{s}_1\notin \mathcal{D}(f_1)$, then we can construct an independent sequence $I_0^{(1)},I_1^{(1)},\dots,I_{k_1}^{(1)}$ with respect to $\mathcal{D}(f_1)$, where $I_{k_1}^{(1)}\subset \mathcal{D}(f_1)$.
	\item Let $f_i\in R_q^{\mathbf{a}}[\mathbf{x}]$ such that $\mathcal{D}(f_{i})=C_{\mathbf{s}_{i-1}}\cup \mathcal{D}(f_{i-1})$. Assume that $\mathbf{s}_i\notin \mathcal{D}(f_i)$, then we can construct an independent sequence $I_0^{(i)},I_1^{(i)},\dots,I_{k_i}^{(i)}$ with respect to $\mathcal{D}(f_i)$, where $I_{k_i}^{(i)}\subset \mathcal{D}(f_i)$.
\end{enumerate}
\begin{corollary}(Shift bound)\label{Minimumdistance}
	With the notations above, if there exists $j$ such that $f_j=0$ and $f_{j-1}\ne 0$, then the minimum distance of $\left\langle f \right\rangle $ is at least $$\min_{0\le i\le j-1}\lvert I_{k_i}^{(i)}\rvert+1.$$
\end{corollary}
In this section, we now present several explicit families of abelian codes obtained via the shift bound. 
In this part,  we present some sufficient conditions. Specifically, by selecting certain specific defining sets and applying the shift bound, we can ensure that the minimum distance of abelian codes is greater than or equal to a designed value.
\begin{theorem}\label{sufficient conditions}
	Let $\mathcal{C}$ be an abelian code over $\mathbb{F}_q$ in the quotient ring $\mathbb{F}_q[x_1,x_2]/\langle x_1^m-1,x_2^n-1 \rangle$ with non-zero dimension, and let $\mathcal{D}$ denote the defining set of $\mathcal{C}$. For each integer $k$ with $0\le k \le n-1$, we define the set $I_k$ as follows:
	\begin{itemize}
		\item If there exist integers $j_k,\delta_k$ satisfying $0\le j_k,\delta_k\le m-1$ such that 
		$
		\{(k,j_k), (k,j_k+1),\dots,(k,j_k+\delta_k-1)\}\subset \mathcal{D}$ and $(k,j_k+\delta_k)\notin \mathcal{D},
		$
		then $I_k=\{(k,j_k), (k,j_k+1),\dots,(k,j_k+\delta_k-1)\}$. Notice that $(k,j_k)$ can be an empty set, that means $(k,j_k)\notin \mathcal{D}$ holds for all integers $j_k$ with $0\le j_k\le m-1$;
		\item Otherwise, $I_k$ is defined to be an infinite set, that means $(k,j_k)\in \mathcal{D}$ holds for all integers $j_k$ with $0\le j_k\le m-1$.
	\end{itemize}
	Then the minimum distance of $\mathcal{C}$ satisfies
	$$
	d(\mathcal{C})\ge \min_{0\le k\le n-1} (k+1)(\lvert I_k\rvert+1).
	$$
	In particular, when $n=1$, the above bound reduces to the BCH bound for cyclic codes.
\end{theorem}
\begin{proof}
	We do not consider the case where $I_k$ is an infinite set, since this has no effect on the lower bound of the minimum distance. We start constructing independent sets from the first row $k=0$. Since $\{(0,j_0), (0,j_0+1),\dots,(0,j_0+\delta_0-1)\}\subset \mathcal{D} \quad \text{and} \quad (0,j_0+\delta_0)\notin \mathcal{D}$ then it is easy to see that $I_0$ is independent with respect to $\mathcal{D}$. Let $\mathcal{D}_{0}^{(1)}=\mathcal{D}\cup \left\lbrace (0,j_0+\delta_0)\right\rbrace $, therefore, $\mathcal{D}_{0}^{(1)}$ has at least $\delta_0+1$ consecutive points in its first row. Suppose that $(0,j_0^{(1)}+\delta_0^{(1)}-1)$ belongs to the set of consecutive points mentioned above, while $(0,j_0^{(1)}+\delta_0^{(1)})\notin \mathcal{D}_{0}^{(1)}$.  Then we can always construct a set $I_0^{(1)}$ that is independent with respect to $\mathcal{D}_{0}^{(1)}$ and satisfies $\lvert I_0^{(1)}\rvert \ge \lvert I_0\rvert$. Repeat the above procedure until $\mathcal{D}_0^{(t_0)}$ covers all points in the first row for some positive integer $t_0$ and we can obtain that  $\lvert I_0^{(t)}\rvert \ge \lvert I_0\rvert$ for all $1\le t \le t_0$. For the second row $k=1$, following the same procedure as above, we can obtain $I_1^{(t)}$ which is independent with respect to $\mathcal{D}_1^{(t)}$, where $1\le t\le t_1$ and $\mathcal{D}_1^{(t_1)}$ covers all points in the second row for some positive integer $t_1$. Take $J_1^{(t)} =  \left( I_1^{(t)}\cup \left\lbrace (1,j_1^{(t)}+\delta_1^{(t)})\right\rbrace  + (-1,0)\right)\cup I_1^{(t)}  $, then $J_1^{(t)}$ is also independent with respect to $\mathcal{D}_1^{(t)}$ and we have $\lvert J_1^{(t)}\rvert=2\lvert I_1^{(t)}\rvert +1\ge 2\lvert I_1\rvert+1$, where $1\le t \le t_k$. We extend the above argument to an arbitrary $1\le k \le n-1$. Take  $J_k^{(t)} =  \cup_{s=1}^{k}\left( I_k^{(t)}\cup \left\lbrace (k,j_k^{(t)}+\delta_k^{(t)})\right\rbrace   + (-s,0)\right)\cup I_k^{(t)}  $, then we have $\lvert J_k^{(t)}\rvert=(k+1)(\lvert I_k^{(t)}\rvert+1) -1\ge (k+1)(\lvert I_k\rvert+1)-1$. From Corollary \ref{Minimumdistance}, we have
	$$d(\mathcal{C})\ge \min_{\begin{matrix}
			0\le k \le n-1\\
			0\le t \le t_k
	\end{matrix}}\left\lbrace \lvert J_k^{(t)}\rvert \right\rbrace+1\ge \min_{0\le k\le n-1} (k+1)(\lvert I_k\rvert+1).$$
\end{proof}
\begin{example}
	We next present an explicit example of an abelian code with minimum distance $\ge 6$. Let $q=9, m=n=q-1=8$. According to the \Cref{sufficient conditions}, it suffices to ensure that there are 5 consecutive points in the first row, 2 consecutive points in the second row , and 1 point on the third to fifth row. The precise defining set is illustrated in  \Cref{Example defining set}.
	\begin{figure}[htbp]
		\centering
		\begin{tikzpicture}[scale=0.5]
			\draw[gray!30] (0,0) grid (7,7);
			\draw[->] (-0.3,0) -- (7.5,0) node[right] {$i$};
			\draw[->] (0,-0.3) -- (0,7.5) node[above] {$j$};
			\foreach \x in {0,1,2,3,4,5,6,7} \draw (\x,-0.1) node[below] {$\x$};
			\foreach \y in {0,1,2,3,4,5,6,7} \draw (-0.1,\y) node[left] {$\y$};
			\fill (0,0) circle (2.5pt);
			\fill (1,0) circle (2.5pt);
			\fill (2,0) circle (2.5pt);
			\fill (3,0) circle (2.5pt);
			\fill (4,0) circle (2.5pt);	
			\fill (0,1) circle (2.5pt);	
			\fill (0,2) circle (2.5pt);
			\fill (0,3) circle (2.5pt);
			\fill (0,4) circle (2.5pt);
			\fill (1,1) circle (2.5pt);
			\node[draw,inner sep=4pt,align=left,anchor=north west] at (9,8) {
				\begin{tabular}{l l}
					\raisebox{1pt}{\tikz{\fill (0,0) circle (2pt);}} &  $(i,j)\in \mathcal{D}$ \\					
				\end{tabular}
			};
		\end{tikzpicture}
		\caption{The distribution of defining set $\mathcal{D}$ of abelian code $[64,54,\ge 6]_9$}
		\label{Example defining set}
	\end{figure}
\end{example}	 

\section{Several families of abelian codes}
In \cite{bibchen},  explicit constructions of cyclic codes with lengths $3(2^r-1)$, $4(3^r-1)$ and so on were proposed, which, however, is restricted to odd integer $r$. In this part, we extend these code lengths to the even $r$ case. We now present the following two cases.
\begin{theorem}\label{First class}
	Let $ m=2^r-1, n=3$, where $r\ge 4$ is an even positive integer.  Then the binary abelian code with defining set $C_{(0,0)}\cup C_{(0,1)} \cup C_{(1,0)}$ has parameters $$\mathcal{C}_1: [3(2^r-1),3(2^r-1)-r-3, 4],$$ and the abelian code with defining set $C_{(0,0)}\cup C_{(0,1)} \cup C_{(1,0)}\cup C_{(1,1)}\cup C_{(3,0)}$ has parameters $$\mathcal{C}_2: [3(2^r-1),3(2^r-1)-3r-3,\ge 6].$$ 
\end{theorem}
\begin{proof}
	It is clear that $\lvert C_{(0,0)} \rvert=1 $, $\lvert C_{(0,1)} \rvert=2$, $\lvert C_{(1,0)}\rvert=\lvert C_{(3,0)} \rvert=r.$ Since $r$ is an even positive integer, then $\lvert C_{(1,1)}\rvert=\text{lcm}(\lvert C_1^{(2^r-1,2)}\rvert,\lvert C_1^{(3,2)}\rvert)=\text{lcm}(r,2)=r$. From \Cref{sufficient conditions}, the lower bounds on the minimum distances of the two codes above can be readily obtained. Suppose that $d(\mathcal{C}_1)\ge 5$, then $$1+\binom{3(2^r-1)}{1}+ \binom{3(2^r-1)}{2}> 2^{r+3}.$$ That shows that the Hamming bound fails for $d(\mathcal{C}_1)\ge 5$.
\end{proof}
\begin{theorem}\label{second class}
	Let $ m=3^r-1, n=4$, where $r\ge 4$ is an even positive integer. Then the ternary abelian code with defining set $C_{(0,0)}\cup C_{(0,1)} \cup C_{(0,2)}\cup C_{(1,0)}\cup C_{(2,0)}\cup C_{(4,0)}\cup C_{(1,1)}$ has parameters $$ [4(3^r-1),4(3^r-1)-4r-4, \ge 6],$$  and the ternary abelian code with defining set $C_{(0,0)}\cup C_{(0,1)} \cup C_{(0,2)}\cup C_{(1,0)}\cup C_{(2,0)}\cup C_{(4,0)}\cup C_{(5,0)}\cup C_{(1,1)}\cup C_{(1,2)}\cup C_{(2,1)}$ has parameters $$ [4(3^r-1),4(3^r-1)-7r-4, \ge 8].$$
\end{theorem}
\begin{proof}
	It is clear that $\lvert C_{(0,0)}\rvert =\lvert C_{(0,2)}\rvert=1, \lvert C_{(0,1)}\rvert=2$ and $\lvert C_{(1,0)}\rvert=\lvert C_{(2,0)}\rvert=\lvert C_{(4,0)}\rvert=\lvert C_{(5,0)}\rvert=\lvert C_{(1,1)}\rvert=\lvert C_{(1,2)}\rvert=\lvert C_{(2,1)}\rvert=r.$ The remaining proof proceeds in a similar manner to that of \Cref{First class}.
\end{proof}
The above two families of abelian codes extend the classical cyclic code constructions reported in \cite{bibchen}, which were originally restricted to odd $r$. The proposed constructions accommodate even $r$, thereby enriching the repository of linear codes with favorable distance parameters. We tabulate their parameters in Table \ref{tab:code-compare}.
\begin{table*}[t]
	\centering
	\footnotesize
	\setlength{\tabcolsep}{4pt} %
	\caption{Comparison of code parameters}
	\label{tab:code-compare}
	\begin{tabular}{c c c c}
		\toprule
		Code & Type & Parameters & Condition \\
		\midrule
		Theorem 1, \cite{bibchen} & Cyclic & $[3(2^r - 1),\,3(2^r - 1)-3r-1,\,6]$ & $r\geq 3$, odd \\
		\Cref{First class} & Abelian & $[3(2^r - 1),\,3(2^r - 1)-3r-3,\,\geq 6]$ & $r\geq 4$, even \\
		Theorem 5, \cite{bibchen} & Cyclic & $[4(3^r - 1),\,4(3^r - 1)-5r-1,\,\ge 6]$ & $r\geq 3$, odd \\
		\Cref{second class} & Abelian & $[4(3^r - 1),\,4(3^r - 1)-4r-4,\,\geq 6]$ & $r\geq 4$, even \\
		Theorem 6, \cite{bibchen} & Cyclic & $[4(3^r - 1),\,4(3^r - 1)-7r-1,\,\ge 8]$ & $r\geq 3$, odd \\
		\Cref{second class} & Abelian & $[4(3^r - 1),\,4(3^r - 1)-7r-4,\,\geq 8]$ & $r\geq 4$, even \\
		\bottomrule
	\end{tabular}
\end{table*}
\begin{theorem}\label{MiniD4}
	Let $q$ be a prime power, $m=n=q+1$, then the abelian code with defining set $C_{(1,0)}\cup C_{(0,1)}\cup C_{(0,0)}$ has parameters
	$$
	[(q+1)^2,(q+1)^2-5, 4]_q.
	$$
\end{theorem}
\begin{proof}
	It is clear that $\lvert C_{(1,0)}\rvert =\lvert C_{(0,1)}\rvert=2$, $\lvert C_{(0,0)}\rvert=1$. From \Cref{sufficient conditions}, the lower bound on the minimum distances of the code is 4. Suppose that the minimum distance is 5, then we have: 
	$$1+(q+1)^2(q-1)+\frac{ (q+1)^2 \left((q+1)^2-1\right)}{2}(q-1)^2> q^5. $$ That shows that the Hamming bound fails for $d\ge 5$.
\end{proof}
In Theorem \ref{MiniD4}, can a zero be added to the abelian code such that the minimum distance of the new code increases by at least 1? That is, whether there is $u$ and $v$ such that the minimum distance of $\left\langle f' \right\rangle $ is at least 5, where $f'$ satisfies $\mathcal{D}(f')=C_{(1,0)}\cup C_{(0,1)}\cup C_{(0,0)}\cup C_{(u,v)}$ and $0\le u,v\le q$. This question is actually can be transformed into the following one:
\begin{question}\label{que}
	Determine the necessary and sufficient conditions on $u$ and $v$
	such that the following homogeneous system: 
	$$	\left\{\begin{matrix}
			x_1+x_2+x_3+x_4=0 \\
			x_1\varepsilon_{q+1}^{i_1}+x_2\varepsilon_{q+1}^{i_2}+x_3\varepsilon_{q+1}^{i_3}+x_4\varepsilon_{q+1}^{i_4}=0 \\
			x_1\varepsilon_{q+1}^{j_1}+x_2\varepsilon_{q+1}^{j_2}+x_3\varepsilon_{q+1}^{j_3}+x_4\varepsilon_{q+1}^{j_4}=0\\
			x_1\varepsilon_{q+1}^{ui_1+vj_1}+x_2\varepsilon_{q+1}^{ui_2+vj_2}+x_3\varepsilon_{q+1}^{ui_3+vj_3}+x_4\varepsilon_{q+1}^{ui_4+vj_4}=0
		\end{matrix}\right.	
	$$has only the trivial solution $(x_1,x_2,x_3,x_4)=(0,0,0,0)$ over $\mathbb{F}_q$, where $0\le i_k,j_k\le q$ and $(i_k,j_k)$ are pairwise distinct, $k=1,2,3,4$.	
\end{question}
Let $q$ be a prime power, and let $\varepsilon=\varepsilon_{q+1}$ be a primitive $(q+1)$‑th root of unity in $\mathbb{F}_{q^2}$. Define
$$
U:=\{z\in\mathbb{F}_{q^2}\mid z^{q+1}=1\}.
$$
Since $0\le i_k,j_k\le q$ exhaust a complete residue system modulo $q+1$, set
$$
a_k=\varepsilon^{i_k},\quad b_k=\varepsilon^{j_k},\qquad P_k=(a_k,b_k)\in U\times U.
$$
Then $P_1,\dots,P_4$ are four distinct points in $U\times U$. The four rows of the system are exactly the evaluations of the monomials $1,\;X,\;Y,\;X^uY^v$ at these four points. The question is therefore transformed into the following.

\begin{question}
	Determine conditions on $u,v$ such that for any four distinct points $P_1,\dots,P_4\in U\times U$,
	$$
	\sum_{k=1}^4x_k=\sum_{k=1}^4x_ka_k=\sum_{k=1}^4x_kb_k=\sum_{k=1}^4x_ka_k^ub_k^v=0,
	$$
	admits only the trivial solution $(x_1,\dots ,x_4)=(0,\dots ,0)$ over $\mathbb{F}_q$.
\end{question} 
\begin{lemma}\label{conjugate}
	Let $\mathbf{x}=(x_1,\dots ,x_4)\in\mathbb{F}_q^4$ satisfy the four equations above. Then $\mathbf{x}$ automatically satisfies
	$$
	\sum_{k}x_ka_k^{-1}=0,\qquad \sum_{k}x_kb_k^{-1}=0,\qquad \sum_{k}x_ka_k^{-u}b_k^{-v}=0.
	$$
\end{lemma}
\begin{proof}
	Take the $q$-th power of the second equation. The Frobenius map is a field homomorphism, $x_k^q=x_k$ since $x_k\in\mathbb{F}_q$, and for $z\in U$, $z^q=z^{q+1}\cdot z^{-1}=z^{-1}$. Thus
	$$
	0=\Bigl(\sum_k x_ka_k\Bigr)^q=\sum_k x_ka_k^{-1}.
	$$
	The remaining identities follow analogously.
\end{proof}
\begin{lemma}\label{Type set}
	Suppose $\mathbf{x}=(x_1,\dots ,x_4)\ne \mathbf{0}$ is a solution. Then
	\begin{itemize}
		\item $x_k\ne 0$ for all $k$;
		\item the multiset $\{a_1,a_2,a_3,a_4\}$ has multiplicity type only $4$, $2+2$, or $1+1+1+1$ (all distinct). The same holds for $\{b_k\}$.
	\end{itemize}
\end{lemma}
\begin{proof}
	Suppose there exists $x_k = 0$, this implies that there exists a codeword of weight less than 4, which contradicts Theorem \ref{MiniD4}. Hence, the assumption is false. Therefore, for every $k$, we have $x_k \neq 0$.
	Type $3+1$: Suppose $a_1=a_2=a_3=\alpha\ne a_4$. From $\sum x_k=0$ and $\alpha(x_1+x_2+x_3)+a_4x_4=0$, we get $(a_4-\alpha)x_4=0$, hence $x_4=0$, contradiction. Type $2+1+1$: Let $a_1=a_2=\alpha$, and $a_3=\beta$, $a_4=\gamma$, with $\alpha,\beta,\gamma$ pairwise distinct. Set $y_1=x_1+x_2$, $y_3=x_3$, $y_4=x_4$. Lemma \ref{conjugate} yields
	$$
	\begin{cases}
		y_1+y_3+y_4 = 0\\
		\alpha y_1+\beta y_3+\gamma y_4 = 0\\
		\alpha^{-1}y_1+\beta^{-1}y_3+\gamma^{-1}y_4 = 0
	\end{cases}.
	$$
	It is easy to obtain that $y_1=y_2=y_3=0$, so $x_3=x_4=0$, contradiction.
\end{proof}
Fix four distinct elements $z_1,\dots,z_4\in\mathbb{F}_{q^2}^*$. Write
$$
Z_k:=\prod_{l\ne k}(z_k-z_l),\qquad F(w;Z):=\sum_{k=1}^4\frac{z_k^{\,w}}{Z_k}\quad(w\in\mathbb Z).
$$
$F(w;Z)$ is the third‑order divided difference of the function $z\mapsto z^w$ at nodes $z_1,\dots,z_4$ and $Z=\left\lbrace z_k:1\le k\le 4\right\rbrace $. If $z_k\in U$, then $z_k^{q+1}=1$, so $F(w;Z)$ depends only on $w\pmod{q+1}$.
\begin{lemma}\label{Fwz}
	With the notations above, we have
	$$
	F(w;Z)=0\ \text{for }w=0,1,2, F(3;Z)=1,$$$$ F(m;Z)=h_{m-3}(z_1,\dots,z_4)\ \text{for }m\ge 3,
	$$
	where $h_r$ denotes the complete homogeneous symmetric polynomial of degree $r$.
\end{lemma}
\begin{proof}
	By Lagrange interpolation, $\sum_k f(z_k)/Z_k$ equals the coefficient of $z^3$ in the interpolating polynomial of degree at most $3$. Hence it vanishes whenever $\deg f\le 2$. For $f(z)=z^m$, since
	$$
	\prod_{l\ne k}\Bigl(1-\frac{z_l}{z_k}\Bigr)=\frac{Z_k}{z_k^3},
	$$
	we have
	$$
	\prod_{l=1}^4\frac1{1-z_lt}=\sum_{k=1}^4\frac{z_k^3/Z_k}{1-z_kt}.
	$$
	Expand both sides as power series in $t$ and compare coefficients of $t^m$. The left‑hand side gives $h_m$, and the right‑hand side gives $\sum_k z_k^{m+3}/Z_k=F(m+3;Z)$.
\end{proof}

\begin{lemma}\label{Fwz2}
	Let $Z=\{z_k\}\subset U$, $Z^{-1}=\{z_k^{-1}\}$, and $e_4=z_1z_2z_3z_4$. For every integer $w$,
	$$
	F(2-w;\,Z^{-1})=-e_4\,F(w;Z).
	$$
	In particular, $F(-1;Z)=-e_4^{-1}\ne 0$, and $F(-n;Z)=-h_{n-1}(Z^{-1})/e_4$.
\end{lemma}
\begin{proof}
	From $z_k^{-1}-z_l^{-1}=-\dfrac{z_k-z_l}{z_kz_l}$,
	$$
	\prod_{l\ne k}(z_k^{-1}-z_l^{-1})
	=\frac{(-1)^3Z_k}{z_k^3\cdot(e_4/z_k)}
	=-\frac{Z_k}{z_k^2e_4}.
	$$
	Then
	$$
	F(2-w;Z^{-1})
	=\sum_k\frac{z_k^{w-2}}{-Z_k/(z_k^2e_4)}
	=-e_4F(w;Z).$$
\end{proof}
\begin{lemma}\label{Solution Lemma}
	Let $a_1,\dots,a_4\in U$ be pairwise distinct. Write $A_k=\prod_{l\ne k}(a_k-a_l)$ and $e_4=\prod_{l=1}^4 a_l$, $A=\left\lbrace a_k:1\le k\le 4\right\rbrace $. Define
	$$
	\eta_k:=\frac{a_k}{A_k},
	\qquad \bm{\xi}:=c\bm{\eta}=c(\eta_1,\dots ,\eta_4),
	$$
	where $c\in\mathbb{F}_{q^2}^*$ satisfies $c^{\,q-1}=-e_4^{-1}$. Then such a $c$ exists, $\bm{\xi}\in\mathbb{F}_q^4$, all components of $\bm{\xi}$ are nonzero, and $V:=\{\mathbf{x}\in\mathbb F_q^4:\ \textstyle\sum x_k=\sum x_ka_k=0\}=\mathbb F_q\bm{\xi}\quad(\dim_{\mathbb F_q}V=1).$
\end{lemma}
\begin{proof}
	\begin{enumerate}[(i)]
		\item Existence of $c$: The image $\{c^{q-1}\mid c\in\mathbb{F}_{q^2}^*\}$ is the subgroup of $\mathbb{F}_{q^2}^*$ of order $q+1$, i.e., $U$. Since $e_4\in U$ and $(-1)^{q+1}=1$, we have $-1\in U$, hence $-e_4^{-1}\in U$.
		\item The vector $\bm{\eta}$ satisfies $\sum\eta_k=F(1,A)=0$, $\sum\eta_ka_k=F(2,A)=0$, and $\sum\eta_ka_k^{-1}=F(0,A)=0$ by Lemma \ref{Fwz}.
		\item Rationality of $\xi$:
		$$
		\eta_k^{\,q}=\frac{a_k^{-1}}{\prod_{l\ne k}(a_k^{-1}-a_l^{-1})}
		=\frac{a_k^{-1}\cdot(-a_k^2e_4)}{A_k}
		=-e_4\,\eta_k.
		$$
		Thus
		$$
		\xi_k^q = c^q(-e_4)\eta_k
		= c\cdot\bigl(c^{q-1}(-e_4)\bigr)\eta_k
		= c\eta_k=\xi_k,
		$$
		so $\bm{\xi}\in\mathbb{F}_q^4$, and clearly $\xi_k\ne0$.
		\item Dimension: $V$ is contained in the kernel over $\mathbb{F}_{q^2}$ of the system given by the three rows $(1,1,1,1)$, $(a_k)$, $(a_k^{-1})$. This $3\times4$ matrix has rank $3$, so its kernel is the one‑dimensional space $\mathbb{F}_{q^2}\bm{\eta}$. Hence $V\subseteq \mathbb{F}_{q^2}\bm{\eta}\cap\mathbb{F}_q^4$. Any $\mathbf{x}=\mu\bm{\eta}\in V$ can be written as $\mathbf{x}=(\mu/c)\bm{\xi}$. Since $\bm{\xi}$ has nonzero components, $\mu/c\in\mathbb{F}_q$, so $V=\mathbb{F}_q\bm{\xi}$.
	\end{enumerate}
\end{proof}

\begin{lemma}\label{Sidon set}
	Suppose $q$ is even. Let $z_1,z_2,z_3,z_4\in U$. If $\{z_1,z_2\}\ne\{z_3,z_4\}$, then $z_1+z_2\ne z_3+z_4$. In particular, the sum of four distinct elements of $U$ is nonzero.
\end{lemma}
\begin{proof}
	Assume $s=z_1+z_2=z_3+z_4$. If $s=0$, then in characteristic two $z_1=z_2$ and $z_3=z_4$. The sets reduce to singletons; this case is excluded in our application because the $z_i$ are distinct, so $s\ne0$. For $s\ne0$, recall $N(x)=x^{q+1}$ and $z^q=z^{-1}$ for $z\in U$. Then
	$$
	N(z_1+z_2)=(z_1+z_2)(z_1^{-1}+z_2^{-1})=\frac{(z_1+z_2)^2}{z_1z_2}.
	$$
	Similarly, $N(z_3+z_4)=s^2/(z_3z_4)$. Equality of norms together with $s\ne0$ yields $z_1z_2=z_3z_4=:\pi$. Then $z_1,z_2$ and $z_3,z_4$ are both roots of $X^2+sX+\pi$, so $\{z_1,z_2\}=\{z_3,z_4\}$, contradiction.
	
	For the corollary: if four distinct elements sum to zero, then $z_1+z_2=z_3+z_4$ in characteristic two, forcing $\{z_1,z_2\}=\{z_3,z_4\}$, contradicting distinctness.
\end{proof}
By Lemma \ref{Type set}, each of $\{a_k\}$ and $\{b_k\}$ can only have multiplicity type: all equal, $2+2$, or all distinct. Using pairwise distinctness of $P_k$, impossible combinations are eliminated:
\begin{itemize}
	\item All‑equal $a$ + all‑equal or $2+2$ $b$: forces coincident points; discard.
	\item $2+2$ for both $a$ and $b$: pairings must cross, giving the rectangle $\{\alpha_1,\alpha_2\}\times\{\beta_1,\beta_2\}$.
\end{itemize}
The remaining cases are listed in \Cref{tab:type_classification}.
\begin{table}[htbp]
	\centering
	\caption{Classification of types.}
	\begin{tabular}{ccc}
		\toprule
		Type  & $\{a_k\}$    & $\{b_k\}$    \\
		\midrule
		I     & all equal    & all distinct \\
		I$'$  & all distinct & all equal    \\
		II    & $2+2$        & $2+2$        \\
		III   & $2+2$        & all distinct \\
		III$'$& all distinct & $2+2$        \\
		IV    & all distinct & all distinct \\
		\bottomrule
	\end{tabular}
	\label{tab:type_classification}
\end{table}
We derive necessary and sufficient conditions for absence of nontrivial solutions for each type.
\subsection*{Type II}
Let $P=(\alpha_m,\beta_n)$ for $m,n\in\{1,2\}$, with $\alpha_1\ne\alpha_2$, $\beta_1\ne\beta_2$. Set $a_1=a_2=\alpha_1, a_3=a_4=\alpha_2, b_1=b_3=\beta_1, b_2=b_4=\beta_2.$ Denote unknowns by $x_{mn}$. Define row sums $r_m=x_{m1}+x_{m2}$. The first two equations give
$$
r_1+r_2=0,\qquad \alpha_1r_1+\alpha_2r_2=0,
$$
so $(\alpha_1-\alpha_2)r_1=0$, hence $r_1=r_2=0$. Consequently
$$
x_{11}=-x_{12}=-x_{21}=x_{22}=:t.
$$
The fourth equation becomes
$$
t\bigl(\alpha_1^u-\alpha_2^u\bigr)\bigl(\beta_1^v-\beta_2^v\bigr)=0.
$$
Since $\alpha_1/\alpha_2$ can be any element of $U\setminus\{1\}$, $\alpha_1^u=\alpha_2^u$ has a solution if and only if $\gcd(u,q+1)>1$. Therefore:

\noindent(N1) Type II is safe $\iff \gcd(u,q+1)=\gcd(v,q+1)=1$.
\subsection*{Type I and I'}
Consider Type I': $b_k\equiv\beta$, $a_k$ pairwise distinct. The third equation holds automatically. By Lemma \ref{Solution Lemma}, any solution is of the form $\bm{x}=\lambda\bm{\xi}$, $\lambda\in\mathbb{F}_q^*$. Substitute into the fourth equation:
$$
\beta^v\sum_k\xi_ka_k^u
= c\lambda\beta^v\sum_k\frac{a_k^{u+1}}{A_k}
= c\lambda\beta^v F(u+1;\{a_k\}).
$$
Type I is symmetric.

Define
$
\mathcal{S}:=\{w\bmod(q+1)\mid F(w;Z)\ne0 \text{ for every 4-element subset }Z\subset U\,\}.
$

\noindent(N2) Types I, I' are safe $\iff u+1\in\mathcal{S}$ and $v+1\in\mathcal{S}$.

From Lemmas \ref{Fwz} and \ref{Fwz2} we immediately obtain basic properties of $\mathcal{S}$:
$$
0,1,2\notin\mathcal{S};\qquad 3,-1\in\mathcal{S};\qquad \mathcal{S}=2-\mathcal{S}.
$$
(The last holds because $Z\mapsto Z^{-1}$ is a bijection on the set of $4$-element subsets of $U$.)
\begin{proposition}
	If $q$ is odd, then (N1) and (N2) cannot hold simultaneously.
\end{proposition}
\begin{proof}
	If $q$ is odd then $2\mid q+1$. Condition (N1) forces $u$ odd, so $w:=u+1$ is even. Choose $a,c\in U$ such that $a\ne\pm c$; this is possible since $-1\in U$, $\lvert U\rvert=q+1\ge4$, and $U$ is partitioned into pairs $\{\pm z\}$ with at least two pairs. Let $Z=\{a,-a,c,-c\}$. Direct computation gives
	$$
	Z_1=\prod_{l\ne1}(a-z_l)=2a(a^2-c^2),$$$$ Z_2=-2a(a^2-c^2),$$$$ Z_3=2c(c^2-a^2),$$$$ Z_4=-2c(c^2-a^2),
	$$
	so
	$$
	F(w;Z)=\frac{a^{w-1}\bigl(1-(-1)^w\bigr)}{2(a^2-c^2)}+\frac{c^{w-1}\bigl(1-(-1)^w\bigr)}{2(c^2-a^2)}=0
	$$
	for even $w$. Hence $u+1\notin\mathcal{S}$.
\end{proof}

From now on set $q=2^e$. By Lemma \ref{Fwz} and \ref{Sidon set}, $4\in\mathcal{S}$. Duality then gives $-2\in\mathcal{S}$. Summarize for even $q$:
$$
\{3,-1,4,-2\}\subseteq\mathcal{S},\qquad \{0,1,2\}\cap\mathcal{S}=\emptyset.
$$
\subsection*{Type III and III'}
Assume $q$ even and (N1) holds. Let $a_1=a_2=\alpha\ne\gamma=a_3=a_4$, with $b_1,\dots,b_4$ pairwise distinct. From the first two equations as in Type II, $x_1+x_2=0$, $x_3+x_4=0$, i.e., $x_2=-x_1$, $x_4=-x_3$, and $x_1,x_3\ne0$.
Set $p=x_1(b_1-b_2)$, $r=x_3(b_3-b_4)$. The third equation yields $p+r=0$. Its conjugate $\sum x_kb_k^{-1}=0$ becomes
$$
-\frac{p}{b_1b_2}-\frac{r}{b_3b_4}=0.
$$
Substitute $r=-p$:
$$
p\biggl(\frac1{b_1b_2}-\frac1{b_3b_4}\biggr)=0.
$$
If $b_1b_2\ne b_3b_4$, then $p=0\Rightarrow x_1=0$, contradiction. Therefore necessarily
$$
b_1b_2=b_3b_4=:\pi,\qquad
x_3=-x_1\frac{b_1-b_2}{b_3-b_4}.
$$
Note that $\dfrac{b_1-b_2}{b_3-b_4}\in\mathbb{F}_q$, since its $q$-th power equals $\dfrac{b_3b_4}{b_1b_2}\cdot\dfrac{b_1-b_2}{b_3-b_4}$.

Substitute into the fourth equation:
$$
x_1\biggl[\alpha^u(b_1^v-b_2^v)-\gamma^u\frac{(b_1-b_2)(b_3^v-b_4^v)}{b_3-b_4}\biggr]=0,
$$
which is equivalent to
$$
\Bigl(\frac{\alpha}{\gamma}\Bigr)^{u}=R:=\frac{(b_1-b_2)(b_3^v-b_4^v)}{(b_3-b_4)(b_1^v-b_2^v)}.
$$
By (N1), $\gcd(v,q+1)=1$, so distinctness of $b_i$ implies distinctness of $b_i^v$, hence the denominator is nonzero.

Using $b_1b_2=b_3b_4=\pi$ together with
$$
(b_1-b_2)^q=\frac{b_1+b_2}{\pi},\qquad
(b_1^v-b_2^v)^q=\frac{b_1^v+b_2^v}{\pi^v},
$$
and analogous identities for $b_3,b_4$, the twisted factors cancel in numerator and denominator, giving $R^q=R$, i.e., $R\in\mathbb{F}_q^*$. On the other hand, $\alpha/\gamma\in U$, so $R\in U\cap\mathbb{F}_q$. For even $q$, $U\cap\mathbb{F}_q=\{z\mid z^2=1\}=\{1\}$, hence $R=1$. Together with $\gcd(u,q+1)=1$ this forces $\alpha=\gamma$, contradiction.

\noindent\textbf{Conclusion.} Under even $q$ and (N1), Types III and III' never admit nontrivial solutions.
\subsection*{Type IV}
Both $\{a_k\}$ and $\{b_k\}$ are sets of four distinct elements. By Lemma \ref{Solution Lemma}, the solution space of the first two equations is $\mathbb{F}_q\xi^{(a)}$, and the solution space of the first and third equations is $\mathbb{F}_q\xi^{(b)}$. Thus the first three equations have a nonzero solution if and only if $\bm{\eta}^{(a)}\parallel\bm{\eta}^{(b)}$ over $\mathbb{F}_{q^2}$, where
$$
\eta^{(a)}_k=\frac{a_k}{A_k},\qquad \eta^{(b)}_k=\frac{b_k}{B_k}.
$$
Call a pair of $4$-tuples $\mathbf{C}=\bigl((a_k),(b_k)\bigr)$ satisfying this condition a collinear configuration. In this situation the fourth equation reads
$$
G(u,v;\mathbf C):=\sum_{k=1}^4\frac{a_k^{\,u+1}b_k^{\,v}}{A_k}=0.
$$

\noindent(N3) Type IV is safe $\iff G(u,v;\mathbf C)\ne0$ for every collinear configuration $\mathbf C$.

Two evident families of collinear configurations. Let $\rho\in U$.
\begin{itemize}
	\item If $b_k=\rho a_k$, then $B_k=\rho^3A_k$, so $\bm{\eta}^{(b)}=\rho^{-2}\bm{\eta}^{(a)}$. Then
	$$
	G=\rho^v\sum_k\frac{a_k^{u+v+1}}{A_k}=\rho^vF(u+v+1;\{a_k\}).
	$$
	\item If $b_k=\rho a_k^{-1}$, then $\bm{\eta}^{(b)}=-e_4\rho^{-2}\bm{\eta}^{(a)}$. Then
	$$
	G=\rho^v\sum_k\frac{a_k^{u-v+1}}{A_k}=\rho^vF(u-v+1;\{a_k\}).
	$$
\end{itemize}
Hence (N3) implies
$$
\textbf{(N3a)}\qquad u+v+1\in\mathcal{S},\qquad u-v+1\in\mathcal{S}.
$$
\begin{theorem}
	Let $q$ be a prime power and. The homogeneous system in Question \ref{que} has only the trivial solution for every collection of pairwise distinct $(i_k,j_k)$ if and only if
	\begin{enumerate}
		\item[(N1)] $\gcd(u,q+1)=\gcd(v,q+1)=1$;
		\item[(N2)] $u+1\in\mathcal{S}$ and $v+1\in\mathcal{S}$;
		\item[(N3)] $G(u,v;\mathcal C)\ne0$ for every collinear configuration $\mathcal C$.
	\end{enumerate}
	Moreover, (N1)$+$(N2) already force $q$ to be even. For even $q$ under (N1), Types II, III, III' are automatically safe, so the three conditions above cover all six  types.
\end{theorem}

\begin{corollary}
	From $0,1,2\notin\mathcal{S}$ together with (N2) and (N3a):
	$$
	u, v, u+v, u-v \not\equiv\ 0, \pm1 \pmod{q+1}, $$$$\gcd(u,q+1)=\gcd(v,q+1)=1,\qquad q=2^e.
	$$
\end{corollary}

\begin{corollary}\label{Small q}
	\begin{itemize}
		\item $q=2$: $q+1=3$. The condition $u\not\equiv0,\pm1\pmod 3$ has no solution. No pair $(u,v)$ exists.
		\item $q=4$: $q+1=5$. We require $u,v\in\{2,3\}$. All sums $u+v$ lie in $\{4,5,6\}\equiv\{4,0,1\pmod 5$, which are excluded. No pair $(u,v)$ exists.
		\item $q=8$: $q+1=9$. Admissible residues: $u,v\in\{2,4,5,7\}$, with $u\pm v\not\equiv0,\pm1\pmod 9$. This leaves
		\begin{align*}
			(u,v)\in\bigl\{&(2,4),(4,2),(2,5),(5,2), \\
			&(4,7),(7,4),(5,7),(7,5)\bigr\}.
		\end{align*}
		Up to the symmetries $\langle u\leftrightarrow v;\;(u,v)\mapsto(-u,-v)\rangle$, there are two essential equivalence classes $\{2,4\}$ and $\{2,5\}$.
	\end{itemize}
\end{corollary}
\begin{example}\label{ex 4}
	Let $q=8$, $m=n=q+1=9$ and $f\in R_{q}^{\mathbf{a}}[\mathbf{x}]$ such that $\mathcal{D}(f)=\left\lbrace (0,0),(1,0),(8,0),(0,1),(0,8) \right\rbrace $, where $\mathbf{a}=(m,n)$. Then according to Theorem \ref{MiniD4}, $\left\langle f \right\rangle $ is a $[81,76,4]_8$ abelian code. Let $g_1,g_2\in R_{q}^{\mathbf{a}}[\mathbf{x}]$ such that $\mathcal{D}(g_1)=\left\lbrace (2,4),(7,5) \right\rbrace$ and $\mathcal{D}(g_2)=\left\lbrace (2,5),(7,4) \right\rbrace$. Then, according to Corollary \ref{Small q} and Magma's calculation, we obtain that $\left\langle fg_1 \right\rangle $ and $\left\langle fg_1g_2 \right\rangle $ are $[81,74,5]_8$, $\textbf{[81,72,6]}_8$ abelian codes, respectively. Since $\left\langle fg_1g_2 \right\rangle \subset \left\langle fg_1 \right\rangle \subset \left\langle f \right\rangle $, then $\textbf{[83,74,6]}_8$ and $\textbf{[83,76,5]}_8$ linear codes can be obtained by Construction X. The codes marked in bold are the record-breaking codes compared with Grassl's code tables \cite{bib5}. We give the distribution of defining set $\mathcal{D}$ of $\textbf{[81,72,6]}_8$ in \Cref{Fig 1}.
\end{example}
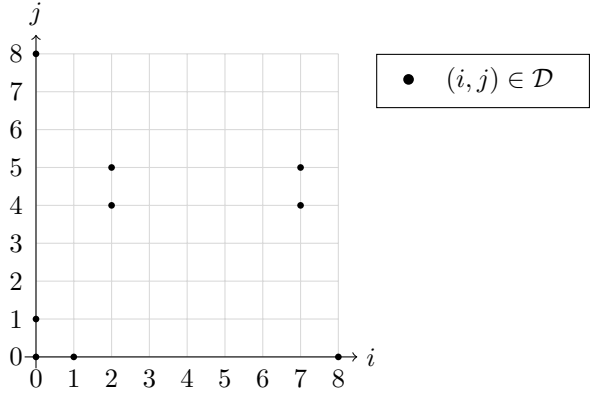
\begin{figure}[htbp]
	\centering
	\begin{tikzpicture}[scale=0.5]
		\draw[gray!30] (0,0) grid (8,8);
		\draw[->] (-0.3,0) -- (8.5,0) node[right] {$i$};
		\draw[->] (0,-0.3) -- (0,8.5) node[above] {$j$};
		\foreach \x in {0,1,2,3,4,5,6,7,8} \draw (\x,-0.1) node[below] {$\x$};
		\foreach \y in {0,1,2,3,4,5,6,7,8} \draw (-0.1,\y) node[left] {$\y$};
		
		\fill (0,0) circle (2.5pt);
		\fill (1,0) circle (2.5pt);
		\fill (0,1) circle (2.5pt);
		\fill (0,8) circle (2.5pt);
		\fill (8,0) circle (2.5pt);	
		\fill (2,5) circle (2.5pt);	
		\fill (7,4) circle (2.5pt);
		\fill (2,4) circle (2.5pt);
		\fill (7,5) circle (2.5pt);
		\node[draw,inner sep=4pt,align=left,anchor=north west] at (9,8) {
			\begin{tabular}{l l}
				\raisebox{1pt}{\tikz{\fill (0,0) circle (2pt);}} &  $(i,j)\in \mathcal{D}$ \\					
			\end{tabular}
		};
	\end{tikzpicture}
	\caption{The distribution of defining set $\mathcal{D}$ of abelian code $[81,72,6]_8$}
	\label{Fig 1}
\end{figure}
\section{Genetic algorithms for constructing record-breaking  abelian codes}
The theoretical constructions presented in Sections 4 and 5 yield several infinite families of abelian codes with proven good parameters. In	 particular, Example \ref{ex 4} tells us that there are some record-breaking linear codes that are abelian codes. However, it is hard to find a universal method to calculate the shift bound of any abelian code. Our strategy is to use the genetic algorithm to search for whether there are more record-breaking abelian codes.
\subsection{Genetic algorithms}
We encode the core algebraic object—\textbf{the defining set} $\mathcal{D}$ directly. For an abelian group $G\cong \mathsf{C}_{a_1}\times \dots \times \mathsf{C}_{a_r}$, we compute all $r$-dimensional $q$-cyclotomic cosets modulo $\mathbf{a}=(a_1,\dots ,a_r)$. A chromosome is a binary vector $(b_1,b_2,\dots ,b_N)$, where $N$ is the number of cosets. $b_i=1$  indicates that the entire $i$-th coset is included in $\mathcal{D}$ (i.e., these are roots of the generator) and $b_i=0$  indicates that the entire $i$-th coset is excluded from $\mathcal{D}$. Each chromosome corresponds to an abelian code.
\begin{example}
	Let us illustrate the chromosome representation with an abelian code. Let $R_2^{(3,3)}=\mathbb{F}_q[x_1,x_2]/<x_1^3-1,x_2^3-1>$. The set of exponents is $\mathcal{N}=\left\lbrace (i,j) \mid (i,j)\in \mathbb{Z}_3\times \mathbb{Z}_3 \right\rbrace $. Compute 2-cyclotomic cosets modulo $(3,3)$: 
	$$C_{(0,0)}=\left\lbrace (0,0)\right\rbrace,$$
	$$C_{(1,0)}=\left\lbrace (1,0),(2,0)\right\rbrace,$$
	$$C_{(0,1)}=\left\lbrace (0,1),(0,2)\right\rbrace,$$
	$$C_{(1,1)}=\left\lbrace (1,1),(2,2)\right\rbrace,$$
	$$C_{(1,2)}=\left\lbrace (1,2),(2,1)\right\rbrace.$$
	We have $N=5$ distinct cosets. A chromosome is defined as a binary vector of length 5: $(b_1,b_2,b_3,b_4,b_5)$, where each bit corresponds to one coset in a fixed order. From Theorem \ref{MiniD4}, the parameters of the abelian code corresponding to the chromosome $(1,1,1,0,0)$ are $[9, 4, 4]_2$. This example demonstrates that the binary chromosome is a precise algebraic blueprint for an abelian code.
\end{example}
The fitness function must predict a code’s proximity to a record-breaking parameter set, we use the following function:
$$\text{Fitness}(\mathbf{b})=d(\mathcal{C}_\mathbf{b})-\texttt{BKLCLowerBound(GF(q),n,k)},$$ where $\mathbf{b}$ is the chromosome, $\mathcal{C}_\mathbf{b}$ is the abelian code corresponding to $\mathbf{b}$, $d(\mathcal{C}_\mathbf{b})$ represents the ture minimum distance of $\mathcal{C}_{\mathbf{b}}$, \texttt{BKLCLowerBound(GF(q),n,k)} is the best known lower bound on the maximum possible minimum weight of
a linear code over finite field $\mathbb{F}_q$ having length $n=\text{Length}(\mathcal{C}_{\mathbf{b}})$ and dimension $k=\text{Dim}(\mathcal{C}_{\mathbf{b}})$. 
\\ \indent To efficiently explore the vast space of abelian codes defined by their defining sets, we employ the following genetic operators:
\begin{itemize}
	\item Selection: We use tournament selection to promote chromosomes with higher fitness. A subset of the population is randomly chosen, and the fittest individual within that subset is selected for reproduction. This balances selection pressure with population diversity.
	\item Crossover: Single-point crossover is applied with a probability $p_{\text{cr}}$. Given two parent chromosomes $\mathbf{p}_1, \mathbf{p}_2$,  a crossover point is chosen randomly. The offspring inherits the first segment from $\mathbf{p}_1$ and the second from  $\mathbf{p}_2$, ensuring that the new chromosome still represents a union of entire $q$-cyclotomic cosets.
	\item Mutation: Each bit in a chromosome is flipped with a small probability $p_{\text{mu}}$. In essence, a single-bit mutation corresponds to adding or removing an entire $q$-cyclotomic coset from the defining set. The mutant code is a subcode of the parent code if a bit changes from 0 to 1, else the mutant code is a supercode of the parent code.
\end{itemize}
To control the chromosome weight (number of 1s) in the Genetic Algorithm in order to leverage Magma’s fast distance computation for codes with very low or very high dimension, we can implement following targeted strategies: Select the chromosomes with very few 1s (e.g., weight $\approx  10–30\%$ of the number of cosets $N$), these correspond to high-dimension codes(or their duals with low-dimension). \\ \indent The complete genetic algorithm procedure is outlined in Algorithm~\ref{alg:ga_abelian}. The algorithm begins with precomputation of $q$-cyclotomic cosets, followed by initialization of a population of random chromosomes. The main evolutionary loop then iterates through fitness evaluation, selection, crossover, and mutation operations. The fitness function $\text{Fitness}(\mathbf{b})$ quantifies how close a code's parameters are to the theoretical optimum. 
\begin{algorithm}
	\caption{Genetic Algorithm for Constructing Abelian Codes}
	\label{alg:ga_abelian}
	\begin{algorithmic}[1]
		
		\Require Finite field $\mathbb{F}_q$, abelian group $G \cong \mathsf{C}_{a_1} \times \cdots \times \mathsf{C}_{a_r}$, population size $P$, max generations $T$, crossover probability $p_{\text{cr}}$, mutation probability $p_{\text{m}}$, tournament size $t$.
		\Ensure Best found abelian code $\mathcal{C}^*$ with parameters $[n,k,d]_q$.
		
		\State \textbf{Step 1: Precomputation}
		\State Compute all $q$-cyclotomic cosets modulo $\mathbf{a} = (a_1,\dots,a_r)$.
		\State Let $N \gets$ number of distinct cosets.
		\State Set chromosome length $= N$.
		
		\State \textbf{Step 2: Initialization}
		\State Initialize empty population $\mathcal{P}$.
		\For{$i = 1$ to $P$}
		\State Randomly generate binary vector $\mathbf{b}_i$ of length $N$ with weight $\approx 10\%-30\%$ (or $70\%-90\%$ for dual search).
		\State $\mathcal{P} \gets \mathcal{P} \cup \{\mathbf{b}_i\}$.
		\EndFor
		
		\State \textbf{Step 3: Main Loop}
		\For{$gen = 1$ to $T$}
		\State \textbf{Fitness Evaluation}
		\For{each $\mathbf{b} \in \mathcal{P}$}
		\State Construct abelian code $\mathcal{C}_{\mathbf{b}}$ from defining set $\mathcal{D}$ encoded by $\mathbf{b}$.
		\State Compute true minimum distance $d(\mathcal{C}_{\mathbf{b}})$ using Magma.
		\State Compute $k = \dim(\mathcal{C}_{\mathbf{b}})$, $n =\mid G\mid$.
		\State $LB \gets \text{BKLCLowerBound}(\mathbb{F}_q, n, k)$ \Comment{e.g., from Grassl's tables}
		\State $\text{fitness}(\mathbf{b}) \gets d(\mathcal{C}_{\mathbf{b}}) - LB$.
		\EndFor
		
		\If{any $\text{fitness}(\mathbf{b}) > 0$} 
		\State Record $\mathbf{b}$ as a record-breaking code.
		\EndIf
		\algstore{breakpoint}
	\end{algorithmic}
\end{algorithm}

\begin{algorithm}
	\caption{Genetic Algorithm for Constructing Abelian Codes (cont.)}
	\begin{algorithmic}[1]
		\algrestore{breakpoint}
		
		\State \textbf{Selection}
		\State Initialize empty mating pool $\mathcal{M}$.
		\While{$\mid \mathcal{M}\mid< P$}
		\State Randomly select $t$ individuals from $\mathcal{P}$.
		\State Choose the fittest among them.
		\State Add selected individual to $\mathcal{M}$.
		\EndWhile
		
		\State \textbf{Crossover}
		\State Initialize empty offspring set $\mathcal{O}$.
		\For{each pair $(\mathbf{p}_1, \mathbf{p}_2)$ in $\mathcal{M}$ (sequential pairs)}
		\If{$\text{rand}() < p_{\text{cr}}$}
		\State Choose random crossover point $cp \in \{1,\dots,N-1\}$.
		\State $\mathbf{o}_1 \gets \text{concat}(\mathbf{p}_1[1:cp], \mathbf{p}_2[cp+1:N])$.
		\State $\mathbf{o}_2 \gets \text{concat}(\mathbf{p}_2[1:cp], \mathbf{p}_1[cp+1:N])$.
		\State $\mathcal{O} \gets \mathcal{O} \cup \{\mathbf{o}_1, \mathbf{o}_2\}$.
		\Else
		\State $\mathcal{O} \gets \mathcal{O} \cup \{\mathbf{p}_1, \mathbf{p}_2\}$.
		\EndIf
		\EndFor	
		\State \textbf{Mutation}
		\For{each $\mathbf{o} \in \mathcal{O}$}
		\For{each bit $j = 1$ to $N$}
		\If{$\text{rand}() < p_{\text{m}}$}
		\State Flip bit $\mathbf{o}[j]$.
		\EndIf
		\EndFor
		\EndFor
		
		\State $\mathcal{P} \gets \mathcal{O}$ \Comment{Replace old population with offspring}
		
		\EndFor
		
		\State \textbf{Step 4: Output}
		\State Return best found code $\mathcal{C}^*$ with highest fitness.
		
	\end{algorithmic}
\end{algorithm}
\subsection{Experimental Results}
We implemented Algorithm~\ref{alg:ga_abelian} in Magma and conducted experiments across various parameter settings. We list our search results as follows:
\\ \indent Let $\mathbf{a} = (m, n)$. We compute all $q$-cyclotomic cosets modulo $\mathbf{a}$. For $\mathbf{s} = (s_1, s_2) \in \mathbb{Z}_{m} \times \mathbb{Z}_n$, the coset is defined as:
$$
C_{\mathbf{s}} = \{(s_1 q^t \bmod m, s_2 q^t \bmod n) \mid t \geq 0\}.
$$

Through systematic computation, we obtain $N = 39$ distinct  4-cyclotomic cosets for $m=15,n=5,q=4$ and $N=59$ distinct  3-cyclotomic cosets for $m=13,n=13,q=3$.  These cosets are enumerated in Table~\ref{tab:cosets-mapping} and Table~\ref{tab:cosets-mapping2}, showing the minimal representative and the complete set of elements for each coset.

\begin{table}[htbp]
	\centering
	\caption{4-Cyclotomic cosets for $\mathbf{a} = (15,5)$ over $\mathrm{GF}(4)$}
	\label{tab:cosets-mapping}
	\begin{tabular}{c c l}
		\toprule
		\textbf{Coset Index} & \textbf{Minimal Rep.} & \textbf{Complete Coset} \\
		\midrule
		1 & $(0,0)$ & $\{(0,0)\}$ \\
		2 & $(0,1)$ & $\{(0,1), (0,4)\}$ \\
		3 & $(0,2)$ & $\{(0,2), (0,3)\}$ \\
		4 & $(1,0)$ & $\{(1,0), (4,0)\}$ \\
		5 & $(1,1)$ & $\{(1,1), (4,4)\}$ \\
		6 & $(1,2)$ & $\{(1,2), (4,3)\}$ \\
		7 & $(1,3)$ & $\{(1,3), (4,2)\}$ \\
		8 & $(1,4)$ & $\{(1,4), (4,1)\}$ \\
		9 & $(2,0)$ & $\{(2,0), (8,0)\}$ \\
		10 & $(2,1)$ & $\{(2,1), (8,4)\}$ \\
		11 & $(2,2)$ & $\{(2,2), (8,3)\}$ \\
		12 & $(2,3)$ & $\{(2,3), (8,2)\}$ \\
		13 & $(2,4)$ & $\{(2,4), (8,1)\}$ \\
		14 & $(3,0)$ & $\{(3,0), (12,0)\}$ \\
		15 & $(3,1)$ & $\{(3,1), (12,4)\}$ \\
		16 & $(3,2)$ & $\{(3,2), (12,3)\}$ \\
		17 & $(3,3)$ & $\{(3,3), (12,2)\}$ \\
		18 & $(3,4)$ & $\{(3,4), (12,1)\}$ \\
		19 & $(5,0)$ & $\{(5,0)\}$ \\
		20 & $(5,1)$ & $\{(5,1), (5,4)\}$ \\
		21 & $(5,2)$ & $\{(5,2), (5,3)\}$ \\
		22 & $(6,0)$ & $\{(6,0), (9,0)\}$ \\
		23 & $(6,1)$ & $\{(6,1), (9,4)\}$ \\
		24 & $(6,2)$ & $\{(6,2), (9,3)\}$ \\
		25 & $(6,3)$ & $\{(6,3), (9,2)\}$ \\
		26 & $(6,4)$ & $\{(6,4), (9,1)\}$ \\
		27 & $(7,0)$ & $\{(7,0), (13,0)\}$ \\
		28 & $(7,1)$ & $\{(7,1), (13,4)\}$ \\
		29 & $(7,2)$ & $\{(7,2), (13,3)\}$ \\
		30& $(7,3)$ & $\{(7,3), (13,2)\}$ \\
		31 & $(7,4)$ & $\{(7,4), (13,1)\}$ \\
		32 & $(10,0)$ & $\{(10,0)\}$ \\
		33 & $(10,1)$ & $\{(10,1), (10,4)\}$ \\
		34 & $(10,2)$ & $\{(10,2), (10,3)\}$ \\
		35 & $(11,0)$ & $\{(11,0), (14,0)\}$ \\
		36 & $(11,1)$ & $\{(11,1), (14,4)\}$ \\
		37 & $(11,2)$ & $\{(11,2), (14,3)\}$ \\
		38 & $(11,3)$ & $\{(11,3), (14,2)\}$ \\
		39 & $(11,4)$ & $\{(11,4), (14,1)\}$ \\
		\bottomrule
	\end{tabular}
\end{table}
\begin{table}[htbp]
	\centering
	\caption{3-Cyclotomic cosets for $\mathbf{a} = (13,13)$ over $\mathrm{GF}(3)$}
	\label{tab:cosets-mapping2}
	\begin{tabular}{c c l}
		\toprule
		\textbf{Coset Index} & \textbf{Minimal Rep.} & \textbf{Complete Coset} \\
		\midrule
		1 & $(0,0)$ & $\{(0,0)\}$ \\
		2 & $(0,1)$ & $\{(0,1),(0,3),(0,9)\}$ \\
		3 & $(0,2)$ & $\{(0,2),(0,6),(0,5)\}$ \\
		4 & $(0,4)$ & $\{(0,4),(0,12),(0,10)\}$ \\
		5 & $(0,7)$ & $\{(0,7),(0,8),(0,11)\}$ \\
		6 & $(1,0)$ & $\{(1,0),(3,0),(9,0)\}$ \\
		7 & $(1,1)$ & $\{(1,1),(3,3),(9,9)\}$ \\
		8 & $(1,2)$ & $\{(1,2),(3,6),(9,5)\}$ \\
		9 & $(1,3)$ & $\{(1,3),(3,9),(9,1)\}$ \\
		10 & $(1,4)$ & $\{(1,4),(3,12),(9,10)\}$ \\
		11 & $(1,5)$ & $\{(1,5),(3,2),(9,6)\}$ \\
		12 & $(1,6)$ & $\{(1,6),(3,5),(9,2)\}$ \\
		13 & $(1,7)$ & $\{(1,7),(3,8),(9,11)\}$ \\
		14 & $(1,8)$ & $\{(1,8),(3,11),(9,7)\}$ \\
		15 & $(1,9)$ & $\{(1,9),(3,1),(9,3)\}$ \\
		16 & $(1,10)$ & $\{(1,10),(3,4),(9,12)\}$ \\
		17 & $(1,11)$ & $\{(1,11),(3,7),(9,8)\}$ \\
		18 & $(1,12)$ & $\{(1,12),(3,10),(9,4)\}$ \\
		19 & $(2,0)$ & $\{(2,0),(6,0),(5,0)\}$ \\
		20 & $(2,1)$ & $\{(2,1),(6,3),(5,9)\}$ \\
		21 & $(2,2)$ & $\{(2,2),(6,6),(5,5)\}$ \\
		22 & $(2,3)$ & $\{(2,3),(6,9),(5,1)\}$ \\
		23 & $(2,4)$ & $\{(2,4),(6,12),(5,10)\}$ \\
		24 & $(2,5)$ & $\{(2,5),(6,2),(5,6)\}$ \\
		25 & $(2,6)$ & $\{(2,6),(6,5),(5,2)\}$ \\
		26 & $(2,7)$ & $\{(2,7),(6,8),(5,11)\}$ \\
		27 & $(2,8)$ & $\{(2,8),(6,11),(5,7)\}$ \\
		28 & $(2,9)$ & $\{(2,9),(6,1),(5,3)\}$ \\
		29 & $(2,10)$ & $\{(2,10),(6,4),(5,12)\}$ \\
		30 & $(2,11)$ & $\{(2,11),(6,7),(5,8)\}$ \\
		31 & $(2,12)$ & $\{(2,12),(6,10),(5,4)\}$ \\
		32 & $(4,0)$ & $\{(4,0),(12,0),(10,0)\}$ \\
		33 & $(4,1)$ & $\{(4,1),(12,3),(10,9)\}$ \\
		34 & $(4,2)$ & $\{(4,2),(12,6),(10,5)\}$ \\
		35 & $(4,3)$ & $\{(4,3),(12,9),(10,1)\}$ \\
		36 & $(4,4)$ & $\{(4,4),(12,12),(10,10)\}$ \\
		37 & $(4,5)$ & $\{(4,5),(12,2),(10,6)\}$ \\
		38 & $(4,6)$ & $\{(4,6),(12,5),(10,2)\}$ \\
		39 & $(4,7)$ & $\{(4,7),(12,8),(10,11)\}$ \\
		40 & $(4,8)$ & $\{(4,8),(12,11),(10,7)\}$ \\
		41 & $(4,9)$ & $\{(4,9),(12,1),(10,3)\}$ \\
		42 & $(4,10)$ & $\{(4,10),(12,4),(10,12)\}$ \\
		43 & $(4,11)$ & $\{(4,11),(12,7),(10,8)\}$ \\
		44 & $(4,12)$ & $\{(4,12),(12,10),(10,4)\}$ \\
		45 & $(7,0)$ & $\{(7,0),(8,0),(11,0)\}$ \\
		46 & $(7,1)$ & $\{(7,1),(8,3),(11,9)\}$ \\
		47 & $(7,2)$ & $\{(7,2),(8,6),(11,5)\}$ \\
		48 & $(7,3)$ & $\{(7,3),(8,9),(11,1)\}$ \\
		49 & $(7,4)$ & $\{(7,4),(8,12),(11,10)\}$ \\
		50 & $(7,5)$ & $\{(7,5),(8,2),(11,6)\}$ \\
		51 & $(7,6)$ & $\{(7,6),(8,5),(11,2)\}$ \\
		52 & $(7,7)$ & $\{(7,7),(8,8),(11,11)\}$ \\
		53 & $(7,8)$ & $\{(7,8),(8,11),(11,7)\}$ \\
		54 & $(7,9)$ & $\{(7,9),(8,1),(11,3)\}$ \\
		55 & $(7,10)$ & $\{(7,10),(8,4),(11,12)\}$ \\
		56 & $(7,11)$ & $\{(7,11),(8,7),(11,8)\}$ \\
		57 & $(7,12)$ & $\{(7,12),(8,10),(11,4)\}$ \\
		\bottomrule
	\end{tabular}
\end{table}

A chromosome is defined as a binary vector $\mathbf{b} = (b_1, b_2, \ldots, b_{N})$. Each bit $b_i$ corresponds to one coset from Table~\ref{tab:cosets-mapping} or Table~\ref{tab:cosets-mapping2}  in the specified order:
$$
b_i = 
\begin{cases}
	1 & \text{if the entire $i$-th coset is included in $\mathcal{D}$}, \\
	0 & \text{otherwise}.
\end{cases}
$$

The defining set $\mathcal{D}$ of the corresponding abelian code is then:
$$
\mathcal{D} = \bigcup_{i: b_i = 1} C_i,
$$
where $C_i$ denotes the $i$-th coset from the table.
Table~\ref{tab:record-codes} presents record-breaking linear codes (in bold) and some other codes (they are prepared for Construction X) discovered by our genetic algorithm search, with comparisons to previously best-known parameters from Grassl's tables \cite{bib5}.

\begin{table}[htbp]
	\centering
	\caption{abelian codes with good parameters discovered by genetic algorithm}
	\label{tab:record-codes}
	\begin{tabular}{c c c c c}
		\toprule
		\textbf{ID} & \textbf{Group} & \textbf{Our Code} & \textbf{Best-Known\cite{bib5}} & Impr. \\
		\midrule
		RB1 & $C_{15}\times C_5$ & $[75, 16, 36]_4$ & $[75, 16, 36]_4$ & $+0$ \\
		RB2 & $C_{15}\times C_5$ & $\bm{[75, 17, 35]}_4$ & $[75, 17, 34]_4$ & $+1$ \\
		RB3 & $C_{15}\times C_5$ & $\bm{[75, 22, 30]}_4$ & $[75, 22, 29]_4$ & $+1$ \\
		RB4 & $C_{15}\times C_5$ & $[75, 18, 34]_4$ & $[75, 18, 34]_4$ & $+0$ \\
		RB5 & $C_{15}\times C_5$ & $\bm{[75, 19, 33]}_4$ & $[75, 19, 32]_4$ & $+1$ \\
		RB6 & $C_{15}\times C_5$ & $\bm{[75, 21, 31]}_4$ & $[75, 21, 30]_4$ & $+1$ \\
		RB7 & $C_{15}\times C_5$ & $[75, 22, 29]_4$ & $[75, 22, 29]_4$ & $+0$ \\
		RB8 & $C_{13}\times C_{13}$ & $[169,15,90]_3$ & $[169,15,90]_3$ & $+0$ \\
		RB9 & $C_{13}\times C_{13}$ & $\bm{[169,18,87]}_3$ & $[169,18,84]_3$ & $+3$ \\
		RB10 & $C_{13}\times C_{13}$ & $\bm{[169,21,81]}_3$ & $[169,21,80]_3$ & $+1$ \\
		RB11 & $C_{13}\times C_{13}$ & $\bm{[169,21,81]}_3$ & $[169,21,80]_3$ & $+1$ \\
		RB12 & $C_{13}\times C_{13}$ & $\bm{[169,24,78]}_3$ & $[169,24,75]_3$ & $+3$ \\
		RB13 & $C_{13}\times C_{13}$ & $\bm{[169,21,81]}_3$ & $[169,21,80]_3$ & $+1$ \\
		RB14 & $C_{13}\times C_{13}$ & $\bm{[169,22,79]}_3$ & $[169,22,78]_3$ & $+1$ \\
		RB15 & $C_{13}\times C_{13}$ & $[169,15,90]_3$ & $[169,15,90]_3$ & $+0$ \\
		RB16 & $C_{13}\times C_{13}$ & $[169,16,88]_3$ & $[169,16,89]_3$ & $-1$ \\
		RB17 & $C_{13}\times C_{13}$ & $\bm{[169,19,84]}_3$ & $[169,19,83]_3$ & $+1$ \\
		\bottomrule
	\end{tabular}	
\end{table}
For the codes in Table~\ref{tab:record-codes}, the corresponding chromosomes are listed in Table \ref{tab:chromosomes-complete}. From Table \ref{tab:chromosomes-complete}, we can conclude the following containment relation:
$$ \text{RB1}\subset \text{RB2},~~ \text{RB4}\subset \text{RB5}\subset \text{RB6}\subset \text{RB7}, $$ $$\text{RB8}\subset \text{RB9}\subset \text{RB10}, ~~ \text{RB11}\subset \text{RB12}, ~~ \text{RB13}\subset \text{RB14},$$ $$\text{RB15}\subset \text{RB16}\subset \text{RB17}.$$
\begin{table}[htbp]
	\centering
	\caption{Complete chromosome information}
	\label{tab:chromosomes-complete}
	\begin{tabular}{c l}
		\toprule
		\textbf{ID}  & \textbf{Binary Representation} \\
		\midrule
		RB1  & 110011111101101011111111111101011011111 \\
		RB2  & 110011111101101011011111111101011011111 \\
		RB3  & 011110101111001110110110111111101010101  \\
		RB4  & 101111110111110111111111101000111111010 \\
		RB5  & 001111110111110111111111101000111111010 \\
		RB6  & 001111110111110111111111101000111011010 \\
		RB7  & 001111110111110111011111101000111011010 \\
		RB8  & 111111111111111001111101111111111101101111111111111111111 \\
		RB9  & 111111111101111001111101111111111101101111111111111111111 \\
		RB10 & 111111111101111001111101111111111101101110111111111111111 \\
		RB11 & 111101111111111111111111111101101111111110011101110111111 \\
		RB12 & 111001111111111111111111111101101111111110011101110111111 \\
		RB13 & 111111111111111011111011111101011111011101110111111111111 \\
		RB14 & 011111111111111011111011111101011111011101110111111111111 \\
		RB15 & 101111111111111111101111111111111110111111010111111111111 \\
		RB16 & 001111111111111111101111111111111110111111010111111111111 \\
		RB17 & 001111111111111111101111111110111110111111010111111111111 \\
		\bottomrule
	\end{tabular}
\end{table}
\\ \indent We apply ConstructionX (Lemma \ref{ConstructionX}) to pairs of nested abelian codes from Table \ref{tab:record-codes}, along with suitably chosen auxiliary codes, to obtain new linear codes with improved parameters. The resulting codes which are listed in Table \ref{tab:construction-x} are compared with Grassl’s tables in the last column. Moreover, the basic construction operations provided by Lemma \ref{PSE} (puncturing, shortening, and extending), we can also obtain many more good codes from the abelian codes presented here, some of which may break existing records. We do not list them all.
\\ \indent To promote reproducibility and facilitate further research, the generator matrices of all record-breaking linear codes and related codes can be accessed at \url{https://github.com/Yucong1234/Abelian-codes}.
\begin{table}[t!]  
	\centering
	\small
	\caption{Construction X examples using abelian codes and auxiliary codes}
	\label{tab:construction-x}
	\begin{tabular}{c c c c c}
		\toprule
		\textbf{Supercode} & \textbf{Subcode} & \textbf{Aux.} & \textbf{Final Code} & \textbf{Best-Known\cite{bib5}} \\
		\midrule
		RB2:$[75,17,35]_4$ & RB1:$[75,16,36]_4$ & $[1,1,1]_4$ & $\bm{[76,17,36]}_4$ & $[76,17,35]_4$  \\
		RB5:$[75,19,33]_4$ & RB4:$[75,18,34]_4$ & $[1,1,1]_4$ & $\bm{[76,19,34]}_4$ & $[76,19,33]_4$  \\
		RB6:$[75,21,31]_4$ & RB5:$[75,19,33]_4$ & $[3,2,2]_4$ & $\bm{[78,21,33]}_4$ & $[78,21,32]_4$  \\
		RB7:$[75,22,29]_4$ & RB6:$[75,21,31]_4$ & $[2,1,2]_4$ & $\bm{[77,22,31]}_4$ & $[77,22,30]_4$  \\
		RB9:$[169,18,87]_3$ & RB8:$[169,15,90]_3$ & $[6,3,3]_3$ & $\bm{[175,18,90]}_3$ & $[175,18,88]_3$  \\
		RB10:$[169,21,81]_3$ & RB9:$[169,18,87]_3$ & $[9,3,6]_3$ & $\bm{[178,21,87]}_3$ & $[178,21,85]_3$  \\
		RB12:$[169,24,78]_3$ & RB11:$[169,21,81]_3$ & $[6,3,3]_3$ & $\bm{[175,24,81]}_3$ & $[178,24,79]_3$  \\
		RB14:$[169,22,79]_3$ & RB13:$[169,21,81]_3$ & $[2,1,2]_3$ & $\bm{[171,22,81]}_3$ & $[171,22,80]_3$  \\
		RB17:$[169,19,84]_3$ & RB15:$[169,15,90]_3$ & $[10,4,6]_3$ & $\bm{[179,19,90]}_3$ & $[179,19,89]_3$  \\
		RB17:$[169,19,84]_3$ & RB16:$[169,16,88]_3$ & $[7,3,4]_3$ & $\bm{[176,19,88]}_3$ & $[176,19,87]_3$  \\
		\bottomrule
	\end{tabular}
\end{table}
\section{Conclusion}
This paper presented new constructions of abelian codes using the shift bound. We obtained several families of abelian codes with good parameters, and further improved some codes using Construction X. Additionally, we introduced a genetic algorithm approach to search for more good codes. Through this method, we discovered several record-breaking linear codes. The genetic algorithm proved effective in exploring the large search space of defining sets and identifying codes with optimal parameters. The results show that genetic algorithms are a useful tool for finding new good codes, and that abelian codes remain a promising area for future research. Both theoretical constructions and computational searches contribute to advancing the understanding of linear codes with good parameters. A further question raised by this work is whether the lower bounds of the record-breaking linear codes found by the genetic algorithm can also be obtained via the shift bound.
\section*{Acknowledgments}
 We thank the anonymous reviewers for their constructive comments, which have significantly improved the quality of this paper. 

This work was supported in part by the National Natural Science Foundation of China under Grant Nos. 62032009, 12171134, U21A20428 and 12571573. This work has been submitted to the IEEE Transactions on Information Theory for possible publication.

\end{document}